\documentclass[aps,10pt,onecolumn,showpacs,citeautoscript,amsmath,amssymb,floatfix,superscriptaddress,pra]{revtex4-2}

\usepackage{amsfonts}
\usepackage{braket}
\usepackage[utf8]{inputenc}
\usepackage[T1]{fontenc}

\usepackage[normalem]{ulem}

\usepackage[colorlinks=true]{hyperref}
\hypersetup{%
  colorlinks = true,
  linkcolor  = black
}
\usepackage{cleveref}
    \crefname{figure}{Figure}{figures}
\usepackage{graphicx}
\usepackage[dvipsnames]{xcolor}
\usepackage{soul}
\usepackage{todonotes}
\usepackage{amsthm}
\usepackage{comment}
\usepackage{dsfont}
\usepackage{tikz}
\usetikzlibrary{calc}
\usetikzlibrary{positioning}
\usepackage{subcaption}

\usepackage{enumitem}

\newtheorem{lemma}{Lemma}
\newtheorem{proposition}{Proposition}
\newtheorem{corol}{Corollary}

\newtheorem{theorem}{Theorem}
\newtheorem{definition}{Definition}
\newtheorem{example}{Example}

\def\be{\begin{equation}}
\def\ee{\end{equation}}
\def\ba{\begin{eqnarray}}
\def\ea{\end{eqnarray}}

\def\a{{\cal A}}

\def\g{{\cal G}}

\def\Nl{{\mathchoice
{\setbox0=\hbox{$\displaystyle\rm N$}\hbox{\hbox to0pt
{\kern0.4\wd0\vrule height0.9\ht0\hss}\box0}}
{\setbox0=\hbox{$\textstyle\rm N$}\hbox{\hbox to0pt
{\kern0.4\wd0\vrule height0.9\ht0\hss}\box0}}
{\setbox0=\hbox{$\scriptstyle\rm N$}\hbox{\hbox to0pt
{\kern0.4\wd0\vrule height0.9\ht0\hss}\box0}}
{\setbox0=\hbox{$\scriptscriptstyle\rm N$}\hbox{\hbox to0pt
{\kern0.4\wd0\vrule height0.9\ht0\hss}\box0}}}}
\def\Zl{{\mathchoice
{\setbox0=\hbox{$\displaystyle\rm Z$}\hbox{\hbox to0pt
{\kern0.4\wd0\vrule height0.9\ht0\hss}\box0}}
{\setbox0=\hbox{$\textstyle\rm Z$}\hbox{\hbox to0pt
{\kern0.4\wd0\vrule height0.9\ht0\hss}\box0}}
{\setbox0=\hbox{$\scriptstyle\rm Z$}\hbox{\hbox to0pt
{\kern0.4\wd0\vrule height0.9\ht0\hss}\box0}}
{\setbox0=\hbox{$\scriptscriptstyle\rm Z$}\hbox{\hbox to0pt
{\kern0.4\wd0\vrule height0.9\ht0\hss}\box0}}}}
\def\Ql{{\mathchoice
{\setbox0=\hbox{$\displaystyle\rm Q$}\hbox{\hbox to0pt
{\kern0.4\wd0\vrule height0.9\ht0\hss}\box0}}
{\setbox0=\hbox{$\textstyle\rm Q$}\hbox{\hbox to0pt
{\kern0.4\wd0\vrule height0.9\ht0\hss}\box0}}
{\setbox0=\hbox{$\scriptstyle\rm Q$}\hbox{\hbox to0pt
{\kern0.4\wd0\vrule height0.9\ht0\hss}\box0}}
{\setbox0=\hbox{$\scriptscriptstyle\rm Q$}\hbox{\hbox to0pt
{\kern0.4\wd0\vrule height0.9\ht0\hss}\box0}}}}
\def\Rl{{\mathchoice
{\setbox0=\hbox{$\displaystyle\rm R$}\hbox{\hbox to0pt
{\kern0.4\wd0\vrule height0.9\ht0\hss}\box0}}
{\setbox0=\hbox{$\textstyle\rm R$}\hbox{\hbox to0pt
{\kern0.4\wd0\vrule height0.9\ht0\hss}\box0}}
{\setbox0=\hbox{$\scriptstyle\rm R$}\hbox{\hbox to0pt
{\kern0.4\wd0\vrule height0.9\ht0\hss}\box0}}
{\setbox0=\hbox{$\scriptscriptstyle\rm R$}\hbox{\hbox to0pt
{\kern0.4\wd0\vrule height0.9\ht0\hss}\box0}}}}
\def\Cl{{\mathchoice
{\setbox0=\hbox{$\displaystyle\rm C$}\hbox{\hbox to0pt
{\kern0.4\wd0\vrule height0.9\ht0\hss}\box0}}
{\setbox0=\hbox{$\textstyle\rm C$}\hbox{\hbox to0pt
{\kern0.4\wd0\vrule height0.9\ht0\hss}\box0}}
{\setbox0=\hbox{$\scriptstyle\rm C$}\hbox{\hbox to0pt
{\kern0.4\wd0\vrule height0.9\ht0\hss}\box0}}
{\setbox0=\hbox{$\scriptscriptstyle\rm C$}\hbox{\hbox to0pt
{\kern0.4\wd0\vrule height0.9\ht0\hss}\box0}}}}
\def\Hl{{\mathchoice
{\setbox0=\hbox{$\displaystyle\rm H$}\hbox{\hbox to0pt
{\kern0.4\wd0\vrule height0.9\ht0\hss}\box0}}
{\setbox0=\hbox{$\textstyle\rm H$}\hbox{\hbox to0pt
{\kern0.4\wd0\vrule height0.9\ht0\hss}\box0}}
{\setbox0=\hbox{$\scriptstyle\rm H$}\hbox{\hbox to0pt
{\kern0.4\wd0\vrule height0.9\ht0\hss}\box0}}
{\setbox0=\hbox{$\scriptscriptstyle\rm H$}\hbox{\hbox to0pt
{\kern0.4\wd0\vrule height0.9\ht0\hss}\box0}}}}
\def\Ol{{\mathchoice
{\setbox0=\hbox{$\displaystyle\rm O$}\hbox{\hbox to0pt
{\kern0.4\wd0\vrule height0.9\ht0\hss}\box0}}
{\setbox0=\hbox{$\textstyle\rm O$}\hbox{\hbox to0pt
{\kern0.4\wd0\vrule height0.9\ht0\hss}\box0}}
{\setbox0=\hbox{$\scriptstyle\rm O$}\hbox{\hbox to0pt
{\kern0.4\wd0\vrule height0.9\ht0\hss}\box0}}
{\setbox0=\hbox{$\scriptscriptstyle\rm O$}\hbox{\hbox to0pt
{\kern0.4\wd0\vrule height0.9\ht0\hss}\box0}}}}
\newcommand{\ca}{\mathcal A}

\newcommand{\cb}{\mathcal B}
\newcommand{\cc}{\mathcal C}
\newcommand{\cd}{\mathcal D}

\newcommand{\cg}{\mathcal G}
\newcommand{\ch}{\mathcal H}
\newcommand{\ci}{\mathcal I}
\newcommand{\cj}{\mathcal J}
\newcommand{\ck}{\mathcal K}

\newcommand{\calp}{\mathcal P}

\newcommand{\cs}{\mathcal S}

\newcommand{\eqa}{\begin{eqnarray}}
\newcommand{\neqa}{\end{eqnarray}}

\def\om{\omega}

\definecolor{myblue}{rgb}{0.2,0.2,0.8}

\usepackage{bm}
\usepackage{bbm}
\def\C{{\mathbbm C}}

\newcommand{\ketbra}[2] {
	| #1 \rangle \! \langle #2 |}

\def\I{{\mathbf 1}}
\def\A{{\sf A}}
\def\B{{\sf B}}
\def\C{{\sf C}}

\definecolor{darkgreen}{rgb}{0.0, 0.5, 0.13}

\begin{document}

\date{\today}

\title{Gravity-mediated entanglement via infinite-dimensional systems}

\author{Stefan L. Ludescher}
\email{Stefan.Ludescher@oeaw.ac.at}
\affiliation{Institute for Quantum Optics and Quantum Information,
Austrian Academy of Sciences, Boltzmanngasse 3, A-1090 Vienna, Austria}
\affiliation{Vienna Center for Quantum Science and Technology (VCQ),
Faculty of Physics, University of Vienna, Vienna, Austria}

\author{Leon D. Loveridge}
\affiliation{Department of Science and Industry Systems, University of South-Eastern
Norway, Kongsberg, 3616, Norway}

\author{Thomas D. Galley}
\affiliation{Institute for Quantum Optics and Quantum Information,
Austrian Academy of Sciences, Boltzmanngasse 3, A-1090 Vienna, Austria}
\affiliation{Vienna Center for Quantum Science and Technology (VCQ),
Faculty of Physics, University of Vienna, Vienna, Austria}

\author{Markus P. M\"uller}
\affiliation{Institute for Quantum Optics and Quantum Information,
Austrian Academy of Sciences, Boltzmanngasse 3, A-1090 Vienna, Austria}
\affiliation{Vienna Center for Quantum Science and Technology (VCQ),
Faculty of Physics, University of Vienna, Vienna, Austria}
\affiliation{Perimeter Institute for Theoretical Physics,
31 Caroline Street North, Waterloo, Ontario N2L 2Y5, Canada}

\begin{abstract}
There has been a wave of recent interest in detecting the quantum nature of gravity with tabletop experiments that witness gravitationally mediated entanglement. Central to these proposals is the assumption that any mediator capable of generating entanglement must itself be nonclassical. However, previous arguments for this have modelled classical mediators as finite, discrete systems such as bits, which excludes physically relevant continuous and infinite-dimensional systems such as those of classical mechanics and field theory. In this work, we close this gap by modelling classical systems as commutative unital $C^*$-algebras, arguably encompassing all potentially physically relevant classical systems.
We show that these systems cannot mediate entanglement between two quantum systems $A$ and $B$, even if $A$ and $B$ are themselves infinite-dimensional or described by arbitrary unital $C^*$-algebras (as in quantum field theory), composed with an arbitrary $C^*$-tensor product. This result reinforces the conclusion that the observation of gravity-induced entanglement would require the gravitational field to possess inherently non-classical features.
\end{abstract}

\maketitle

\section{Introduction}

Proposals for tabletop experiments that witness the non-classical nature of gravity have recently attracted significant interest~\cite{schneider2022,bose2025massive}. These proposed experiments take place in low-energy regimes and therefore cannot determine which, if any, specific quantum gravity theory is correct, however they can rule out certain classical theories of gravity. In this paper, we will focus on gravity-mediated entanglement experiments~\cite{bose2017,marletto2017}, which seek to rule out classical theories of gravity based on the observation of entanglement between two masses induced by gravitational interaction.

In these experimental proposals, two initially unentangled masses interact solely via the gravitational field. When describing this interaction using the Newtonian potential~\cite{bose2017,marletto2017}, it has been shown that entanglement is generated over time. Detecting entanglement would rule out certain models of semi-classical gravity~\cite{schneider2022}. The most prominent approach to semi-classical gravity is Møller-Rosenfeld semi-classical gravity~\cite{moller1962theories,rosenfeld1963quantization} which is based on the semi-classical Einstein equation and leads to a non-linear Schr{\" o}dinger evolution of the quantum system~\cite{bahrami2014schrodinger} (see also~\cite{anastopoulos2014problems}). This model is not operationally consistent~\cite{mielnik1980mobility,gisin1989stochastic,gisin1990weinberg}; however, there exist consistent models of classical gravity interacting with quantum matter~\cite{oppenheim2023postquantum,tilloy2024}, though some of these, such as the Diosi-Penrose model~\cite{diosi1989models,penrose1996gravity}, predict entanglement mediation via a non-local mechanism~\cite{trillo2025diosi} and hence would not be ruled out by the observation of GME.

A key assumption in the inference of the non-classical nature of the gravitational field from the observation of GME is that the gravitational field acts locally. However, it has been demonstrated that, in the low-energy approximation, the creation of entanglement can be traced back to the Newtonian potential part of the Hamiltonian rather than the dynamical part associated with gravitons~\cite{schneider2022, fragkos2022}. In this approximation, the Newtonian potential acts non-locally on the two masses, in contradiction with the assumption that the mediator (i.e., gravity) acts locally on the two masses (see also~\cite{hall2021comment,martinmartinez2023gravity,marchese2025newton} for further criticisms of the assumptions of the GME experiments). Nevertheless, it has been argued that, based on other experiments, we already know that gravity itself is a field that acts locally, and therefore, modelling the influence of the gravitational field as non-local is merely an artifact of the approximation~\cite{christodoulou2023b}. Indeed it has been shown that the predictions of GME derived using the non-local Newtonian potential can also be recovered modelling the gravitational field quantum-mechanically and interacting locally with the matter systems~\cite{bose2022mechanism,christodoulou2023a}.

Beyond the inference of the non-classical nature of the gravitational field, the detection of GME has been argued to have other implications. For example, it has been argued that it would imply that the gravitational field must be in a superposition of geometries during the experiment~\cite{schneider2022, christodoulou2019}, or that virtual gravitons have been exchanged~\cite{marshman2020locality,DanielsonSatishchandranWald2022}.

The existing literature can be divided into approaches which explicitly model the gravitational field as a classical or quantum system coupling with a specific interaction to the matter degrees of freedom, and those which are model-independent, and make general inferences about the nature of the gravitational field subject to some assumptions.

Model-independent approaches either treat the gravitational field as a classical mediator, as in~\cite{marletto2017,galley2022}, or as a classical communication channel as in~\cite{bose2017,altamirano2018gravity}. In~\cite{marletto2017} the field is modelled as a classical bit, whilst in~\cite{galley2022} it is modelled as a discrete classical system of arbitrary finite size. In the present work, we significantly generalize the classical mediator-based approaches and model the gravitational field as an infinite-dimensional classical system. In the appendices, we explicitly translate the LOCC protocol into a classical mediator approach and show that our approach provides a significant generalization of the LOCC based approach too. Moreover, we explicitly allow the two quantum systems that are to be entangled to be infinite-dimensional as well, which includes operator-algebraic subsystems as they appear in quantum field theory.

Although considering a finite version of a problem can often give us a good first intuition and can be an interesting case of study, the gravitational field is an infinite-dimensional object, both in general relativity and in existing models of semi-classical gravity. As such, generalizing the existing GME arguments to the infinite-dimensional case is a necessary step in being able to use them to rule out classical gravity.

\section{Classical and Quantum physics in the language of $C^*$-algebras}
Before diving into the analysis of gravity-mediated entanglement experiments, we will recall how classical and quantum physics can be formulated within the language of $C^*$-algebras. Roughly speaking, a $C^*$-algebra can be thought of as a subalgebra of the algebra $\cb(\ch)$ of bounded operators on some Hilbert space $\ch$. In the classical case, the $C^*$-algebra is commutative, i.e.\ all of its elements commute with each other. For a brief reminder of some basic definitions and properties of $C^*$-algebras see the appendix.

We give a brief outline of how classical and quantum physics can be described using $C^*$-algebras, based on~\cite{landsman2017}, to which we refer for a more detailed presentation. We start with classical physics. Consider a finite, discrete classical configuration space $X$ (in the infinite case, $X$ could represent the phase space of a classical mechanical system or of a classical field theory). The set of all functions from $X$ to $\mathbb{C}$ is denoted as $C(X)$. By defining multiplication of a function with a scalar and addition of two functions in $C(X)$ pointwise, $C(X)$ becomes a vector space (which is in this case isomorphic to $\mathbb{C}^n$, with $n$ the cardinality of $X$). By introducing multiplication of two functions pointwise, $C(X)$ becomes an algebra. Physical properties of a system are described by real numbers. Thus, the physically relevant functions are the real-valued functions $R(X)$. Clearly, $R(X)$ can be seen as a real-linear subspace of $C(X)$. To describe $R(X)$ in algebraic terms, we introduce the involution $f^*(x) = \overline{f(x)}$, which turns $C(X)$, together with the sup norm $\|f\|_\infty = \sup_{x \in X}  |f(x)| $, into a $C^*$-algebra. The subspace $R(X)$ coincides with the set of self-adjoint functions $C(X)_{\mbox{\small{sa}}}$, i.e., the set of functions for which $f^*=f$. The observable algebra of a classical system configuration space $X$ is then given by the commutative $C^*$-algebra $C(X)$.

In classical physics, the states of a system are the probability measures over the configuration (or phase) space $X$. In the finite case, a state is a function $p:\calp(X)\rightarrow [0,1]$ such that $p(X)=1$ and $p(A\cup B)=p(A)+p(B)$ for all disjoint sets $A,B\in\calp(X)$, where $\calp(X)$ denotes the power set of $X$. The state space of a classical system is then given by the space $\mbox{Pr}(X)$ of all probability measures on $X$, and we note that $\mbox{Pr}(X)$ is convex and compact, where convexity ensures that a probabilistic mixture of two states is itself a state.

For general $C^*$-algebras, states are defined as normalized positive linear functionals on the algebra. Hence, in the classical case, a state is a functional $\omega:C(X)\rightarrow \mathbb{C}$ s.t. $\omega(\mathbf{1}_X)=1$, where $\mathbf{1}_X(x)=1$ for all $x\in X$ and $\omega(f^*f)\geq0$ for all $f\in C(X)$. The linear functional $\mathbf{1}_X$ is the unit element of $C(X)$, i.e.\ $\mathbf{1}_X\cdot f=f\cdot\mathbf{1}_X=f$, and a $C^*$-algebra containing a unit element is called unital. In the following we will explore the connection to the previous definition of states as probability measures on phase space.
We first note that $C(X)$ is the complexification of $R(X)$. The condition that a functional $\omega_R:R(X)\rightarrow\mathbb{R}$ is positive translates to $\omega(f)\geq 0$ if $f(x)\geq 0$ for every $x\in X$. Now let $\omega_R$ be a state on $R(X)$. We can extend $\omega_R$ to a complex-linear functional $\omega:C(X)\rightarrow\mathbb{C}$ by $\omega(f+ig)=\omega_R(f)+i\omega_R(g)$, where $f,g\in R(X)$ and we recall that we can split every $h\in C(X)$ into $h=f+ig$ for some $f\in R(X)$ and some $g\in R(X)$. It is easy to check then that $\omega$ is a state on $C(X)$. On the other hand, restricting an arbitrary state $\omega$ in $C(X)$ to $R(X)$ will give a state on $R(X)$. Therefore we can think of states on $C(X)$ and $R(X)$ as the same object $\omega$. Now we see that every element of the state space $\cs(C(X))$ can be associated with an element of $\mbox{Pr(X)}$ and vice versa. We start with a state $\omega\in\cs(C(X))$ and define the function $P:X\rightarrow\mathbb{R}$ by $P(x)=\omega(\delta_x)$. Furthermore, we have $\sum_{x\in X}\delta_x=\mathbf{1}_X$ and thus $\sum_{x\in X}P(x)=1$, and therefore, $P$ is a probability distribution. By setting $p(U)=\sum_{x\in U}P(x)$ for $U\in\calp(X)$, we find the corresponding probability measure $p$ which is generated by $\omega$. Conversely, consider a probability measure $p\in\mbox{Pr}(X)$ and define $\omega:R(X)\rightarrow \mathbb{R}$ by $\omega(f)=\sum_{x\in X} p(x)f(x)$. Obviously, $\omega$ is positive and normalized and therefore a state on $R(X)$ and thus on $C(X)$.

 Infinite-dimensional commutative $C^*$-algebras can also be understood as functions on some underlying configuration space $X$. Let us denote the set of all continuous functions on $X$ that vanish at infinity by $C_0(X)$. It is well known that every commutative $C^*$-algebra is isomorphic to $C_0(X)$ for some locally compact Hausdorff space $X$, which includes a large variety of topological spaces, for example all manifolds.

Moreover, even in the infinite-dimensional case, we can associate the states on $C_0(X)$ with probability measures. More precisely, let $\omega$ be a state on $C_0(X)$. Then there exists a unique probability measure $\mu$ on $X$ such that $\omega(f)=\int_{X}d\mu(x)f(x)$ for every $f\in C_0(X)$. On the other hand, every probability measure $\mu$ on $X$ defines a state $\omega$ on $C_0(X)$ by the same equation as before. Thus, we can use commutative $C^*$-algebras to model general finite- and infinite-dimensional classical systems.

\begin{example}\label{ex:inf_classical_system}
    As a well-known example of a unital commutative $C^*$-algebra, we consider the space of bounded sequences $\ell^\infty$ (see e.g.~\cite{Weaver2017}), defined by
    \begin{equation*}
        \ell^\infty:=\left\{(x_n)_{n\in\mathbb{N}}\in\mathbb{C}^{\mathbb{N}}\;\left|\; \sup_{n\in\mathbb{N}} |x_n|<\infty \right.\right\}.
    \end{equation*}
    The space is equipped with the norm $\|x\|_{\infty}=\|(x_n)_{n\in\mathbb{N}}\|_\infty=\sup_{n\in\mathbb{N}}|x_n|$. We will see that we can interpret $\ell^\infty$ as an infinite-dimensional generalization of diagonal matrices. To this end, we consider the Hilbert space of square-summable sequences $\ell^2$, which is defined by
    \begin{equation*}
        \ell^2:=\left\{(x_n)_{n\in\mathbb{N}}\in\mathbb{C}^{\mathbb{N}}\; \left| \;\sum_{n\in\mathbb{N}}|x_n|^2<\infty\right.\right\}.
    \end{equation*}
    The inner product is defined by $\braket{x|y}:=\sum_{n\in \mathbb{N}}\overline{x_n}y_n$. We will write for the standard orthonormal basis $\{e_n\}_{n\in\mathbb{N}}$:\begin{align*}
        \ket{1}&=e_1=(1,0,0,\ldots),\\
        \ket{2}&=e_2=(0,1,0,\ldots),\\
        &\;\;\vdots
    \end{align*}
    We can embed $(\ell^\infty,\|\cdot\|_\infty)$ into $(\cb(\ell^2),\|\cdot\|_{\cb(\ell^2)})$, the bounded operators on the Hilbert space $\ell^2$, where $\|\cdot\|_{\cb(\ell^2)}$ is the operator norm, with the map $\iota:\ell^\infty\hookrightarrow\cb(\ell^2)$:\begin{equation*}
        (x_n)_{n\in\mathbb{N}}\mapsto \sum_{n\in\mathbb{N}} x_n\ketbra{n}{n}.
    \end{equation*}
Since $\ell^\infty$ is a commutative unital $C^*$-algebra, it is isomorphic to the set of continuous functions on some compact space $X$. Indeed, $\ell^\infty$ is isomorphic to the set of continuous functions $C( \beta\mathbb{N})$ on the Stone-\v{C}ech-compactification $\beta\mathbb{N}$ of the natural numbers $\mathbb{N}$ (see e.g.~\cite{Aviles2025}). 
\end{example}
This example illustrates that a given commutative $C^*$-algebra may be interpreted as a classical physical system in several different ways: $\ell^\infty$ can be seen as the algebra of bounded functions on the configuration space $\mathbb{N}$, or as the algebra of continuous functions on the compact configuration space $\beta\mathbb{N}$. This observation is relevant for the interpretation of our main theorem. That is, when considering a classical system (say, gravity) that potentially mediates entanglement, we will model the system in terms of a unital $C^*$-algebra $\mathcal{G}$. Due to unitality, $\mathcal{G}$ can be interpreted as the algebra of continuous functions on some compact Hausdorff space $X$. At first sight, it seems as if this excludes classical mechanical systems with a \textit{non-compact} phase space of $\mathbb{R}^{2n}$ (systems of finitely many particles), or classical fields (where $X$ is infinite-dimensional and not compact). However, as the example above shows, unital $C^*$-algebras encompass more general classical systems, such as those of bounded measurable functions on some non-compact 
phase space. For the case of mechanics of finitely many particles, this has been shown explicitly in~\cite{Duvenhage}, where classical mechanics is described in terms of a $C^*$-algebra of bounded complex-valued Borel functions on phase space $\mathbb{R}^{2n}$.

Now we turn our attention to the non-commutative case. It is easy to check that the algebra of bounded operators $\cb(\ch)$ on some Hilbert space $\ch$ is a $C^*$-algebra. On the other hand, it can be shown~\cite{Gelfand1943,bratteli1987} that every $C^*$-algebra is isomorphic to a norm-closed selfadjoint subalgebra of the bounded operators of some Hilbert space $\ch$. In this sense, $C^*$-algebras can be interpreted as sub-theories of Hilbert space quantum theories.

\section{No mediation of entanglement via classical infinite-dimensional systems}

In this section, we will use the language of $C^*$-algebras to analyze entanglement mediation scenarios involving classical mediators and prove that it is impossible to create entanglement in these kinds of scenarios.

\begin{figure}[h]
\centering
\begin{tikzpicture}[thick, every node/.style={scale=1}]

  \node (A) at (0,0) {$A$};
  \node (G) at (1.5,0) {$G$};
  \node (B) at (3,0) {$B$};

  \node[draw, minimum width=2.8 cm, minimum height=1cm, anchor=north] (TAG) at (0.7,2) {$T_{\mathcal{A}\mathcal{G}}$};
  \node[draw, minimum width=1.4cm, minimum height=1cm, anchor=north] (phiB) at (3,2) {$\phi_{\mathcal{B}}$};

  \node[draw, minimum width=1.4cm, minimum height=1cm, anchor=north] (phiA) at (0,4) {$\phi_{\mathcal{A}}$};
  \node[draw, minimum width=2.8cm, minimum height=1cm, anchor=north] (TGB) at (2.3,4) {$T_{\mathcal{G}\mathcal{B}}$};

  \draw[->] (A) -- ($(TAG.south west)!0.25!(TAG.south east)$);
  \draw[->] (G) -- ($(TAG.south west)!0.78!(TAG.south east)$);
  \draw[->] (B) -- (phiB.south);

  \draw[->] ($(TAG.north west)!0.25!(TAG.north east)$) -- (phiA.south);
  \draw[->] ($(TAG.north west)!0.78!(TAG.north east)$) -- ($(TGB.south west)!0.21!(TGB.south east)$);

  \draw[->] (phiB.north) -- ($(TGB.south west)!0.75!(TGB.south east)$);

  \draw[->] (phiA.north) -- ++(0,0.8);
  \draw[->] ($(TGB.north west)!0.21!(TGB.north east)$) -- ++(0,0.8);
  \draw[->] ($(TGB.north west)!0.75!(TGB.north east)$) -- ++(0,0.8);

\end{tikzpicture}
\caption{We consider two quantum systems $A$ and $B$, and a classical mediator $G$. At each step, either $A$ or $B$ can interact with the mediator $G$,
but not both at the same time. $A$ and $B$ are described by arbitrary unital $C^*$-algebras, and $\mathcal{G}$ by an arbitrary unital commutative $C^*$-algebra, and they are composed with an arbitrary $C^*$-tensor product. This includes the cases where $A$ and $B$ are described by the algebras of all bounded operators $\mathcal{A} = \mathcal{B}(\mathcal{H}_A)$ and $\mathcal{B} = \mathcal{B}(\mathcal{H}_B)$ on arbitrary infinite-dimensional Hilbert spaces $\mathcal{H}_A$ and $\mathcal{H}_B$, respectively, as well as much more general scenarios where, for example, both algebras $\mathcal{A}$ and $\mathcal{B}$ are von Neumann algebras.}
\label{setup}
\end{figure}

We consider two laboratories in which we prepare two quantum systems $A$ and $B$. Furthermore, the systems $A$ and $B$ can interact via a classical mediator $G$ . The main case of interest is when $G$ is the gravitational field, but our analysis is agnostic to the physical interpretation of $G$. The interaction of the quantum systems with $G$ is assumed to be local; that is, at each layer of interaction, $G$ interacts either with  $A$ or  with $B$ (see Figure~\ref{setup}).

 To model this setup mathematically, we attach unital $C^*$-algebras $\ca$ and $\cb$ to the systems $A$ and $B$, respectively. To describe the combined system of $A$ and $B$, we will consider a tensor product $\ca\otimes\cb$, obtained by completing the algebraic tensor product $\ca\odot\cb$ with respect to an arbitrary but fixed $C^*$-norm $\|\cdot\|$. See the appendix for some more information on the $C^*$-algebraic tensor product.
Furthermore, we will attach a unital commutative $C^*$-algebra $\cg$ to the mediator $G$, which, as discussed above, arguably describes a very general notion of classical system. Note that it is always possible to adjoin an identity element to a non-unital $C^*$-algebra, which will allow the application of our results even in the non-unital case. In contrast to the composite $AB$ of $A$ and $B$, the composite $AGB$ of $AB$ and $G$ is unique: this is because every commutative $C^*$-algebra is nuclear~\cite{takesaki1964}, that is, its tensor product with any other $C^*$-algebra is unique. Overall, the total system is hence described by a tensor product $C^*$-algebra $\ca\otimes\cg\otimes\cb$. In the Conclusions section, we will say more about how this corresponds to scenarios that we encounter in quantum field theory in Minkowski space.

At the beginning of the experiment, the states in the laboratories $A$ and $B$ are prepared independently of each other and of the gravitational field $G$. Thus, the total initial state $\omega_{\ca\cg\cb}\in(\ca\otimes\cg\otimes\cb)^*$ is described by a product state $\omega_{\ca\cg\cb}=\omega_\a\otimes\omega_\cg\otimes\omega_\cb$, which is a special case of a triseparable state:
\begin{definition}
A \textit{triseparable state}~\cite{Eggeling} is any element in the weak-$*$ closure of the set
\[
   \left\{
       \sum_{i=1}^k\lambda_i \omega_\ca^{(i)}\otimes\omega_\cg^{(i)}\otimes\omega_\cb^{(i)}
   \right\},
\]
where every $\omega_\cs^{(i)}$ is a state on $\cs$, $k\in\mathbb{N}$, and $\lambda_i\geq 0$ as well as $\sum_{i=1}^k \lambda_i=1$.
\end{definition}
The proof that no entanglement can be mediated via the classical mediator $G$ consists of two steps. First, we will show that local interactions between the mediator and one of the systems map triseparable states to triseparable states, which will be a consequence of the fact that all states on the tensor product of a general unital $C^*$-algebra and a unital commuting $C^*$-algebra are separable~\cite{takesaki1979} and of the structure of the interactions. Then, in the second step, we will show that the reduced state on $\ca \otimes \cb$ of a triseparable state on $\ca \otimes \cg \otimes \cb$ is separable.

For the first step, we will start by clarifying what we mean by local interaction.
A linear map $\phi:\cc\to\cd$ between unital $C^*$-algebras is called \textit{unital} if $\phi(\mathbf{1_\cc})=\mathbf{1}_\cd$. Furthermore, $\phi$ is completely positive (CP) if\begin{equation}
    \sum_{i,j=1}^n D^*_i\phi(C_i^*C_j)D_j\geq 0,
\end{equation}
for all $n\in\mathbb{N}$, $C_i\in\cc$ and $D_i\in\cd$. Equivalently, the map $\phi$ is CP if $\mbox{id}_{\mathbb{C}^{k\times k}}\otimes\phi:\mathbb{C}^{k\times k}\otimes\cc\rightarrow\mathbb{C}^{k\times k}\otimes\cd$, acting as
\begin{equation}\label{cpmap}
    \begin{pmatrix}
    \cc_{11}&\ldots &\cc_{1k}\\
    \vdots& \ddots& \vdots\\
    \cc_{k1}&\ldots&\cc_{kk}
    \end{pmatrix}\mapsto
    \begin{pmatrix}
    \phi(\cc_{11})&\ldots &\phi(\cc_{1k})\\
    \vdots& \ddots& \vdots\\
    \phi(\cc_{k1})&\ldots&\phi(\cc_{kk})
    \end{pmatrix},
\end{equation}
is positive for every $k\in\mathbb{N}$. In~(\ref{cpmap}), the identification of $\mathbb{C}^{k\times k}\otimes\cc$ with $\cc^{k\times k}$, i.e.\ with the $C^*$-algebra of $k\times k$-matrices with entries $C_{ij}\in\cc$, is made.
A channel $\phi$ (in the Heisenberg picture) from $\mathcal{A}$ to $\mathcal{B}$, where $\mathcal{A}$ and $\mathcal{B}$ are $C^*$-algebras, is then by definition a unital CP map $\phi:\mathcal{B}\to\mathcal{A}$~\cite{OhyaPetz}.
 
It is easy to see that the dual map $\phi^*:\cc^*\rightarrow\cd^*$ of a channel $\phi:\cc\rightarrow\cd$ describes the Schr\"odinger picture of $\phi$ and maps states to states: let $\omega$ be a state on $\cd$, then, for every positive $C\in\cc$, we have $\phi(C)\geq 0$ and thus $\phi^*(\omega)[C]=\omega[\phi(C)]\geq 0$, and 
\begin{equation*}
    \phi^*(\omega)[\mathbf{1_\cc}]=\omega[\phi(\mathbf{1}_\cc)]=\omega[\mathbf{1}_\cd]=1,
\end{equation*}
and therefore $\phi^*(\omega)$ is a state on $\cc$.
Moreover, channels are always bounded~\cite{russo1966note}, and thus, the dual map $\phi^*$ is weak*-continuous (see Proposition 1.3 in Chapter VI of~\cite{conway2007}).

An interaction between $A$ and $G$ can be described by a channel of the form $T_{\ca\cg} : \ca \otimes \cg \rightarrow \ca \otimes \cg$. To ensure that the interaction between $A$ and $G$ occurs locally, we will describe the total channel by a product channel that acts on $A$ and $G$ independently of $B$, that is, by a channel of the form $T_{\ca\cg} \otimes \phi_{\cb}$, where $\phi_\cb : \cb \rightarrow \cb$ is a channel as well. Note that some care has to be taken here: in general, the map $T_{\ca\cg} \otimes \phi_{\cb}$ is only well-defined on the algebraic tensor product $(\ca \otimes \cg) \odot \cb$. In general, additional assumptions may be needed to ensure that there is an extension of this map to a channel on $\ca\otimes\cg \otimes \cb$, and Example~\ref{ExNonexistence} in the appendix shows a case where an extension of this kind does not exist. If a completely positive and thus continuous extension exists, it is unique, since $(\ca\otimes\cg)\odot\cb$ is dense in $\ca\otimes\cg\otimes\cb$. For example, according to~\cite[Proposition 4.23]{takesaki1979}, an extension as a channel exists if the two algebras $\ca\otimes\cg$ and $\cb$ are composed via the minimal or the maximal tensor product. Here, we make the assumption that there exists an extension as a channel, because otherwise, the physical description of our scenario (as in Figure~\ref{setup}) would not even make sense: we certainly only need to worry about the mediation of entanglement by \textit{physically implementable} channels, i.e.\ by $T_{\ca\cg}$ that can actually be implemented locally on $AG$ consistent with the parallel existence of $B$.

\begin{lemma}
\label{TagTrisepToTrisep}
The dual map $(T_{\ca\g}\otimes\phi_\cb)^*$ of a channel of the form $T_{\ca\g}\otimes\phi_\cb$ maps triseparable states to triseparable states on $\ca\otimes\cg\otimes\cb$.
\end{lemma}
\begin{proof}
We start with a product state $\omega_\ca\otimes\omega_\cg\otimes\omega_\cb$. A general element of $Z\in\ca\otimes\cg\otimes\cb$ lies in the norm-closure of the algebraic tensor product $(\ca\otimes\cg)\odot\cb$ of $\a\otimes\cg$ and $\cb$, thus it can be written as the limit of linear combinations of simple tensors between elements of $\ca\otimes\cg$ and $\cb$. That is, we can write every such element $Z$ in the form $Z=\lim_\alpha\sum_{i=1}^{n_\alpha}D_\alpha^i\otimes B^i_\alpha$, where $D_\alpha^i\in\ca\otimes\cg$, $B^i_\alpha\in\cb$, and $\lim_\alpha$ is taken with respect to the norm on $\ca\otimes\cg\otimes\cb$. Then \begin{align}
    ((T_{\ca\cg}\otimes\phi_\cb)^*(\omega_\ca\otimes\omega_\cg\otimes\omega_\cb))[Z]&=(\overbrace{(\omega_\ca\otimes\omega_\cg\otimes\omega_\cb)}^{\mbox{continuous}}\left(\overbrace{(T_{\ca\cg}\otimes\phi_\cb)}^{\mbox{continuous}}\left[\lim_\alpha\sum_{i=1}^{n_\alpha}D_\alpha^i\otimes B^i_\alpha\right]\right)\nonumber\\
    &=\lim_\alpha\sum_{i=1}^{n_\alpha}(\omega_\ca\otimes\omega_\cg\otimes\omega_\cb)\left((T_{\ca\cg}\otimes\phi_\cb)\left[D_\alpha^i\otimes B^i_\alpha\right]\right)\nonumber\\
     &=\lim_\alpha\sum_{i=1}^{n_\alpha}(\omega_\ca\otimes\omega_\cg\otimes\omega_\cb)\left[T_{\ca\cg}(D_\alpha^i)\otimes \phi_\cb(B^i_\alpha)\right]\nonumber\\
     &=\lim_\alpha\sum_{i=1}^{n_\alpha}(\omega_\ca\otimes\omega_\cg)[T_{\ca\cg}(D_\alpha^i)]\cdot\omega_\cb\left[\phi_\cb(B^i_\alpha)\right]\nonumber\\
&=\lim_\alpha\sum_{i=1}^{n_\alpha}T^*_{\ca\cg}(\omega_\ca\otimes\omega_\cg)[D_\alpha^i]\cdot\phi^*_\cb(\omega_\cb)\left[B^i_\alpha\right]\nonumber\\
&=\lim_\alpha\sum_{i=1}^{n_\alpha}\omega_{\ca\cg}[D_\alpha^i]\cdot\tilde\omega_\cb\left[B^i_\alpha\right]=\lim_\alpha\sum_{i=1}^{n_\alpha}\overbrace{\omega_{\ca\cg}\otimes\tilde\omega_\cb}^{\mbox{continuous}}\left[D_\alpha^i\otimes B^i_\alpha\right]\nonumber\\
&=\omega_{\ca\cg}\otimes\tilde\omega_{\cb}\left[\lim_\alpha\sum_{i=1}^{n_\alpha}D_\alpha^i\otimes B^i_\alpha\right]=\omega_{\ca\cg}\otimes\tilde\omega_{\cb}\left[Z\right],
\end{align}
for every $Z\in\ca\otimes\cg\otimes\cb$. We have defined $\omega_{\ca\cg}=T^*_{\ca\cg}(\omega_\a\otimes\omega_\cg)$ and $\tilde\omega_\cb=\phi^*_\cb(\omega_\cb)$, which are states on $\ca\otimes\cg$ and $\cb$, respectively.
Next, we want to show that $\omega_{\ca\cg}\otimes\tilde\omega_{\cb}$ is a triseparable state. In ~\cite{takesaki1979} it was proven that every pure state on the tensor product of a unital commuting (i.e.\ classical) and some other unital $C^*$-algebra is always a product state. Combining this with the fact that the state space of a unital $C^*$-algebra is the weak*-closure of the convex hull of its pure states~\cite{bratteli1987} tells us that in this case, every state $\omega_{\ca\cg}$ is separable.   
Thus we can write $\omega_{\ca\cg}$ as the weak*-limit of a converging net of convex combinations of product states, i.e.\ $\omega_{\ca\cg}[d]=\lim_\beta\sum_{i=1}^{n_\beta}\lambda^i_\beta\left(\omega_\ca\right)^i_\beta\otimes(\omega_\cg)^i_\beta[d]$ for every $d\in\ca\otimes\cg$, where $(\omega_\ca)_\beta^i\in\ca^*$ and $(\omega_\cg)^i_\beta\in\cg^*$ are states and $\lambda_\beta^i\geq0$ and $\sum_{i_1}^{n_\beta}\lambda^i_\beta=1$ for every $\beta$. 

In the appendix, we show that the map $\omega\mapsto \omega\otimes\tau$ is weak*-continuous (see Proposition~\ref{wstartensorpro}). Thus,
\begin{equation}
\omega_{\ca\cg}\otimes\tilde{\omega}_{\cb}[Z]=\left(\left(\lim_\beta\sum_{i=1}^{n_\beta}\lambda^i_{\beta}(\omega_{\ca})^i_\beta\otimes(\omega_\cg)^i_\beta\right)\otimes\tilde\omega_\cb\right)[Z]=\lim_\beta\sum_{i=1}^{n_\beta}\lambda^i_{\beta}\left((\omega_{\ca})^i_\beta\otimes(\omega_\cg)^i_\beta\otimes\tilde\omega_\cb\right)[Z],
    \end{equation}
    for every $Z\in\ca\otimes\cg\otimes\cb$. Therefore, $\omega_{\ca\cg}\otimes\tilde\omega_\cb$ is triseparable.

So far we have seen that $(T_{\ca\g}\otimes\phi_\cb)^*$ maps $\omega_\ca\otimes\omega_\cg\otimes\omega_\cb$ to a triseparable state, but from the linearity and the weak* continuity of $(T_{\ca\cg}\otimes\phi_\cb)^*$ it follows that it maps general triseparable states to triseparable states. 
\end{proof}

By the symmetry of the problem, the dual of a channel of the form $\phi_\ca\otimes T_{\cg\cb}$ will also map triseparable states to triseparable states; see Lemma~\ref{LemAssociativitySortOf} in the appendix for why we are free to regard the scenario as a composite of the form $(AG)B$ or $A(GB)$, despite the fact that $C^*$-tensor products are not in general associative. Furthermore, countable sequences of local interactions between $A$ and $G$, and/or between $B$ and $G$, will map a triseparable state to a triseparable state (if the limit exists):

\begin{corol}\label{cor:sequence_states}
Let $\{T^{(n)}_{\ca\cg\cb}\}_n$ be a sequence of channels where, for each $n$, the channel is either of the form
$T^{(n)}_{\ca\cg\cb} = T^{(n)}_{\ca\cg} \otimes \phi^{(n)}_{\cb}$ or $T^{(n)}_{\ca\cg\cb} = \phi^{(n)}_{\ca} \otimes T^{(n)}_{\cg\cb}$, where $T^{(n)}_{\ca\g},T^{(n)}_{\cg\cb},$ $\phi^{(n)}_{\ca}$ and $\phi^{(n)}_{\cb}$ are also channels.
Define the sequence $\{\tilde{T}^{(n)}_{\ca\cg\cb}\}_n$ by
$\tilde{T}^{(n)}_{\ca\cg\cb} = T^{(n)}_{\ca\cg\cb} \circ \cdots \circ T^{(1)}_{\ca\cg\cb}$. Let $\omega_{\ca\cg\cb}$ be a triseparable state. If the weak*-limit of the sequence $\{(\tilde T^{(n)}_{\ca\cg\cb})^*(\omega_{\ca\cg\cb})\}_n$ exists, then it is a triseparable state.
\end{corol}
\begin{proof}
From Lemma~\ref{TagTrisepToTrisep}, it follows immediately that for every $n$, it holds that $(\tilde{T}^{(n)}_{\ca\cg\cb})^*(\omega_{\a\cg\cb})=(T^{(1)}_{\ca\cg\cb})^*\circ\ldots\circ(T^{(n)}_{\ca\cg\cb})^*(\omega_{\a\cg\cb})$ is a triseparable state. Since the set of the triseparable states is weak*-closed, the weak*-limit is also a triseparable state. 
\end{proof}
At the end of the experiment, we are interested in the reduced state on $\ca\otimes\cb$ and whether it can be entangled. The reduced state of $\omega_{\ca\cb\cg}$ is defined by $\omega_{\ca\cb}(\cdot)=\omega_{\ca\cb\cg}((\cdot)\otimes\mathbf{1}_{\cg})$, where we changed the order of $\ca,\cb$ and $\cg$ for convenience.
\begin{lemma}
    The reduced state $\omega_{\ca\cb}$ of a triseparable state $\omega_{\a\cb\cg}$ is separable.
\end{lemma}
\begin{proof}
    It is easy to check that the map $\phi_\mathbf{1}:\ca\otimes\cb\rightarrow\ca\otimes\cb\otimes\cg$ defined by $D \mapsto D\otimes\mathbf{1}_{\cg}$ is bounded, and therefore, $\phi^*_{\mathbf{1}}$ is weak*-continuous. Consider a product state $\omega_{\ca}\otimes\omega_{\cb}\otimes\omega_{\cg}$, then \begin{equation}
    \omega_{\ca\cb}(D)=\phi^*_\mathbf{1}(\omega_\ca\otimes\omega_\cb\otimes\omega_\cg)(D)=\omega_\ca\otimes\omega_\cb\otimes\omega_\cg(D\otimes\mathbf{1}_\cg)=\omega_\ca\otimes\omega_\cb(D)\omega_\cg(\mathbf{1}_\cg)=\omega_\ca\otimes\omega_\cb(D),
\end{equation}
for all $D\in\ca\otimes\cb$. Therefore, $\omega_{\ca\cb}=\omega_A\otimes\omega_B$, which is a separable state. Furthermore, since $\phi^*_\mathbf{1}$ is linear and weak*-continuous, and the set of separable states is weak*-closed, it follows that a general triseparable state gets mapped to a separable one.
\end{proof}
This concludes our discussion on why a classical system cannot mediate entanglement.

\section{Conclusion}
In this work, we have extended previous no-go results on gravitationally mediated entanglement to the case of continuous classical mediators. By adopting a $C^*$-algebraic formalism, we have shown that classicality, defined in terms of commutativity of the observable algebra, fundamentally prohibits the generation of entanglement between two initially uncorrelated quantum systems. This holds regardless of whether the classical mediator is finite or infinite-dimensional, discrete or continuous, and regardless of the operator-algebraic description of the quantum systems and their choice of tensor product composition. Our analysis generalizes the classical mediator and LOCC-based arguments used in prior studies, which  typically rely on restrictive assumptions such as finite-dimensional or discrete mediators as well as finite-dimensional quantum systems.

It is interesting to compare our results to recent works by van Luijk et al.~\cite{van2024relativisticpub,VL2,VL3}, who consider LOCC protocols for operator-algebraic quantum subsystems of a Hilbert space $\mathcal{H}$. These are modelled as von Neumann subalgebras $\mathcal{M}_A,\mathcal{M}_B$ of the bounded operators of $\mathcal{H}$ in ``Haag duality'', i.e.\ $\mathcal{M}_B$ is exactly the set of operators that commute with all elements of $\mathcal{M}_A$ (denoted $\mathcal{M}_B=\mathcal{M}_A'$), and vice versa. They show (see e.g.~\cite[Theorem D]{VL3}) that LOCC protocols for the corresponding bipartite system can extract an arbitrary amount of entanglement in terms of Bell pairs from any given pure state of an ancillary system if $\mathcal{M}_A$ (and thus $\mathcal{M}_B$) is not of type I. At first sight, this seems to be in tension with our result, but their mathematical framework does not quite capture the scenario of gravity-mediated-entanglement (GME) experiments. This is because, in their commutator-based framework beyond type I, \textit{``all [bipartite] states contain an infinite amount of single-shot entanglement''}~\cite{VL3}, and there is no analog of a product state. That is, it does not make sense to talk about independently locally prepared  systems, which is however a prerequisite of the GME protocol.
\begin{figure}
     \centering
     \begin{subfigure}{0.45\textwidth}
         \centering
         \includegraphics[width=\linewidth]{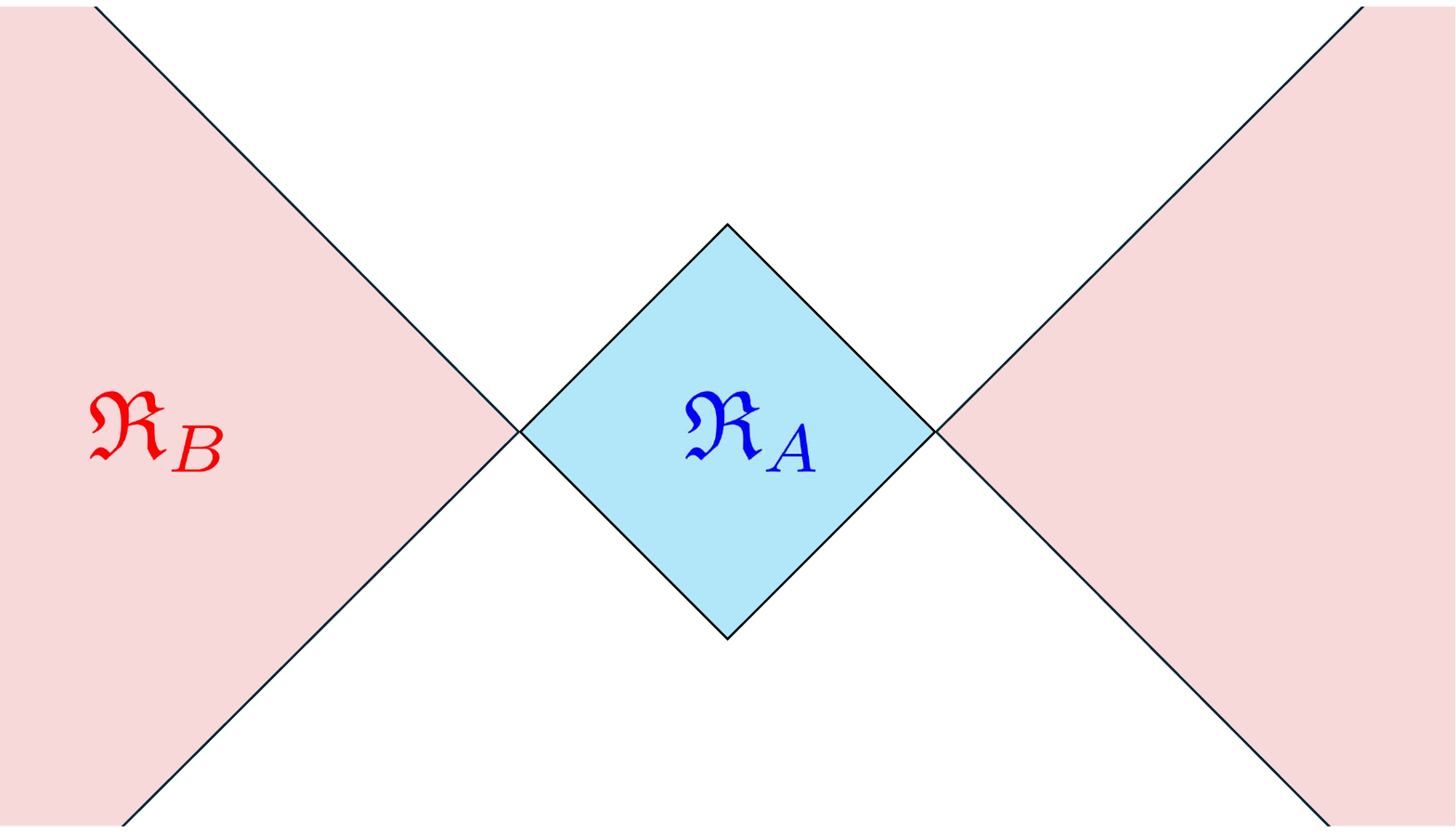}
         \caption{Spacetime regions for Haag duality.}
         \label{haag}
     \end{subfigure}
     \begin{subfigure}{0.45\textwidth}
         \centering
         \includegraphics[width=\linewidth]{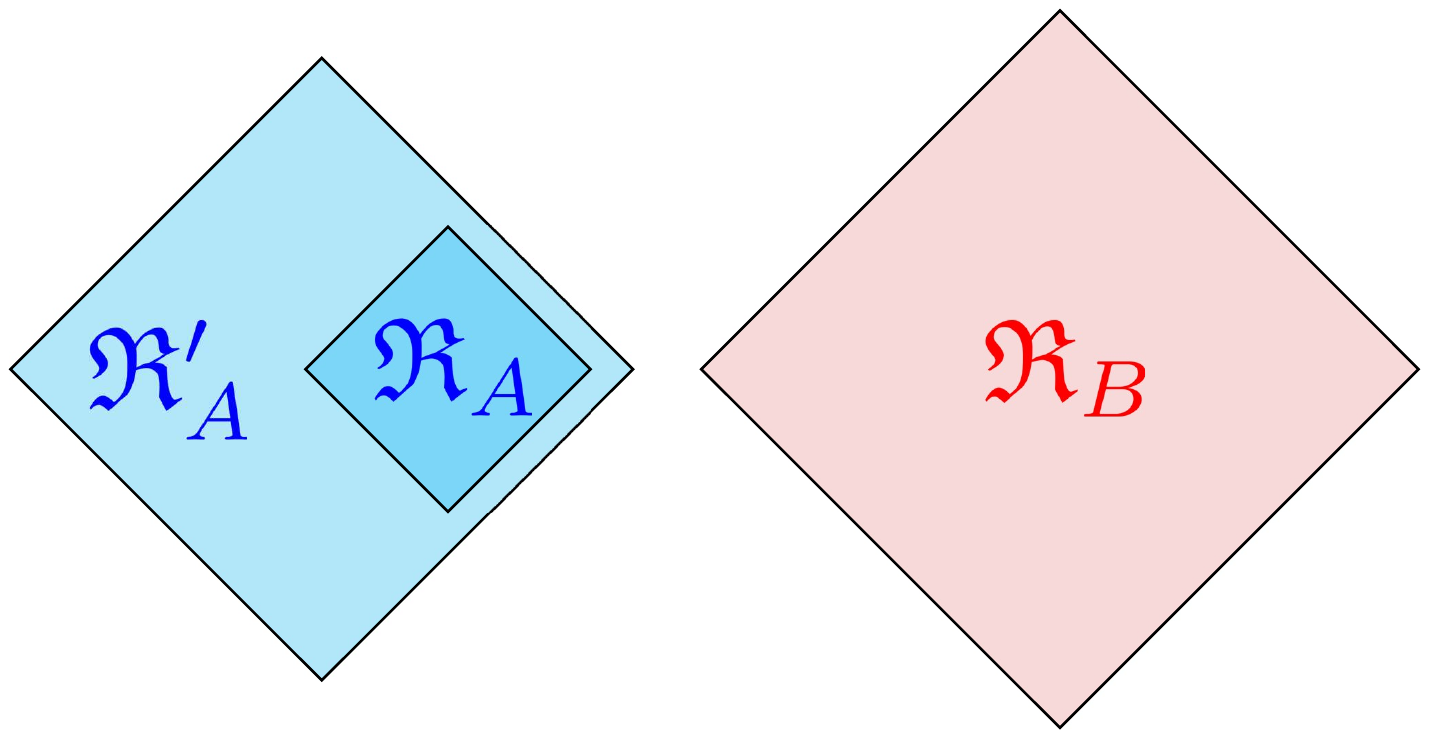}
         \caption{Spacetime regions for the split property. See also \cite{Fewster} for a more detailed discussion.}
         \label{split}
     \end{subfigure}
     \caption{In Fig. \ref{haag}, region $\mathfrak{R}_B$ is the full causal complement of region $\mathfrak{R}_A$, and thus we have Haag duality $\ca(\mathfrak{R}_B)=\ca(\mathfrak{R}_A)'$. In this case, there is no analogue of product states and hence no meaningful notion of independently locally prepared systems. In Fig.~\ref{split}, the regions $\mathfrak{R}_A$ and $\mathfrak{R}_B$ are spacelike separated and they do not touch. We find a region $\mathfrak{R}'_A$ which compactly contains the closure of $\mathfrak{R}_A$ and does not overlap with $\mathfrak{R}_B$. The split property then implies that there exists a notion of locally and independently prepared states on $\mathfrak{R}_A$ and $\mathfrak{R}_B$.}
     \label{spacetimepictures}
\end{figure}

We can turn to quantum field theory for a more concrete picture of the difference of  their setting and ours (see Fig. \ref{spacetimepictures}). To every (open and relatively compact) region $\mathfrak{R}$ of Minkowski spacetime, there is a von Neumann algebra $\mathcal{A}(\mathfrak{R})$ of operators associated with that region. Consider two spacelike separated regions $\mathfrak{R}_A$ and $\mathfrak{R}_B$, then Einstein causality means that $[\mathcal{A}(\mathfrak{R}_A),\mathcal{A}(\mathfrak{R}_B)]=0$. If $\mathfrak{R}_B$ is the full causal complement of $\mathfrak{R}_A$, then we have Haag duality~\cite{Haag}, and $\mathcal{A}(\mathfrak{R}_B)=\mathcal{A}(\mathfrak{R}_A)'$, which reproduces the setting by can Luijk et al. On the other hand, suppose that $\mathfrak{R}_A$ and $\mathfrak{R}_B$ do not touch at their boundaries. This is the situation that we have in the case of two laboratories $A$ and $B$ of bounded spatiotemporal extent and non-zero distance from each other. In this case~\cite{Fewster}, we find a region $\mathfrak{R}'_A$ ``slightly larger'' than $\mathfrak{R}_A$, i.e.\ such that the closure of $\mathfrak{R}_A$ is compactly contained in $\mathfrak{R}'_A$, and $\mathfrak{R}'_A\cap\mathfrak{R}_B=\emptyset$. The \textit{split property} now ensures that there is a type I factor $\mathcal{A}$ such that $\mathcal{A}(\mathfrak{R}_A)\subset\mathcal{A}\subset\mathcal{A}(\mathfrak{R}'_A)$. But this implies that there is a notion of locally independently prepared states on $\mathfrak{R}_A$ and $\mathfrak{R}_B$, and the corresponding global von Neumann algebra $\mathcal{A}(\mathfrak{R}_A)\vee \mathcal{A}(\mathfrak{R}_B)$ is isomorphic to the tensor product $\mathcal{A}(\mathfrak{R}_A)\otimes_{\rm min} \mathcal{A}(\mathfrak{R}_B)$~\cite{brunetti2014}. In this case, it makes sense to talk about locally independently prepared quantum systems, and our conclusion (no GME by classical systems) applies.

It has recently been claimed that  finite-dimensional classical systems can mediate entanglement in the GME setup~\cite{vidal2025bose}. In order to understand how this result relates to our work, we first compare our proof method to that of~\cite{galley2022}, which explicitly required both kinematical and compositional properties of the classical mediating system. Kinematically, it required the existence of a resolution of the identity transformation (i.e.\ the existence of a maximal non-disturbing measurement), and compositionally, it required the standard tensor product composition rule.  Operationally, under the assumptions of the framework of general probabilistic theories (GPTs)~\cite{Mueller2021,Plavala2023}, the use of a tensor product encodes the requirement of \textit{tomographic locality}~\cite{Wootters,Hardy2001,Mueller2021,Plavala2023}: states of a composite system $AB$ are uniquely determined by the statistics and correlations of local measurements on $A$ and on $B$. This principle is satisfied by classical and quantum theory, and by many (but not all) GPTs. The framework of $C^*$-algebras allows for definitions of subsystems which satisfy local tomography~\cite[Definition 1]{vanLuijkSchwonnekStottmeisterWerner2024}, but which do not meet all the other assumptions of subsystems in the GPT framework. In general, the equivalence between local tomography and the existence of a tensor product need not hold in the $C^*$-algebraic framework. However, as shown in~\cite[Theorem 5]{vanLuijkSchwonnekStottmeisterWerner2024}, for $C^*$-algebras, local tomography together with the assumption of \emph{statistical independence} (which is met in GPTs) is equivalent to the existence of a tensor product. This clarifies the operational content of our assumption that the composition of $A$, $G$ and $B$ is described by a $C^*$-tensor product.

Our result is another piece of the puzzle in answering the question of what the observation of entanglement in a GME experiment would imply: the observation of entanglement between spatially separated quantum systems interacting only through gravity would point towards the gravitational field not being adequately described by a classical model (understood as both locally classical and compositionally standard, i.e. composing in a locally tomographic way), even in a very general, infinite-dimensional formulation.

Our findings thus reinforce the view that experiments aiming to detect gravity-induced entanglement offer a powerful probe of the quantum nature of gravity. They also close a significant conceptual loophole in existing arguments, eliminating the possibility that the classicality of the mediator might be salvaged by appealing to infinite-dimensional continuous models. In future work, it would be interesting to see to what extent our no-go result for this \textit{formal} definition of a classical system (as a commutative unital $C^*$-algebra) extends to systems that are classical according to an \textit{operational} definition (for example, admitting of a generalized noncontextual ontological model~\cite{Schmid}). Another natural question is whether and how our results extend to the multipartite case, which has recently been studied as another avenue towards more elaborate tests of gravity-mediated entanglement (see e.g.~\cite{Miki,Ghosal}).

\section*{Acknowledgments}
We are grateful to Andrea Di Biagio
for helpful discussions. This research was funded in part by the Austrian Science Fund (FWF) 10.55776/PAT2839723, via the Austrian Science Fund (FWF) project P 33730-N, and via the Austrian Science Fund (FWF) project COE 1 Quantum Science Austria. Furthermore, this research was supported in part by Perimeter Institute for Theoretical Physics. Research at Perimeter Institute is supported by the Government of Canada through the Department of Innovation, Science, and Economic Development, and by the Province of Ontario through the Ministry of Colleges and Universities.

\appendix
\section{C*-Algebras}
In the appendix, we will give an overview of some definitions and facts about $C^*$-algebras that we need in the main text. For a more detailed introduction to $C^*$-algebras, see e.g.~\cite{bratteli1987, takesaki1979, murphy1990}, and for a more comprehensive presentation of $C^*$-tensor products, see e.g.~\cite{takesaki1979, murphy1990, Bruckler1999}.  
\begin{definition}
    Let $\ca$ be a vector space over the field $\mathbb{K}$. $\ca$ is an algebra if $\ca$ is equipped with a multiplication $\cdot:\ca\times\ca\rightarrow\ca$, denoted as $A\cdot B=AB$, such that\begin{enumerate}
        \item $A(BC)=(AB)C$
        \item $A(B+C)=AB+AC$
        \item $\alpha\beta(AB)=(\alpha A)(\beta B)$
    \end{enumerate}
    for every $A,B,C\in\ca$ and every $\alpha,\beta\in\mathbb{K}$.
\end{definition}
In the following, we will consider algebras over $\mathbb{C}$ unless otherwise stated.
\begin{definition}
Let $\ca$ be an algebra. A map $*:\ca\rightarrow\ca$, denoted as $*(A)=A^*$, is called an involution, or adjoint operator, if
\begin{enumerate}
    \item $A^{**}=A$,
    \item $(AB)^*=B^*A^*$,
    \item $(\alpha A+\beta B)^*=\bar\alpha A^*+\bar\beta B^*$
\end{enumerate}
for all $A,B\in\ca$ and all $\alpha,\beta\in\mathbb{C}$.
\end{definition}
An algebra equipped with an involution is called a *-algebra. A \textit{normed algebra} is an algebra with a norm $\|\cdot\|$ that is submultiplicative, $\|AB\|\leq\|A\|\cdot\|B\|$. Such a norm induces a metric topology on $\ca$, which is called the uniform/norm topology. 
\begin{definition}
    A normed *-algebra $\ca$ which is complete with respect to the norm and satisfies $\|A\|=\|A^*\|$ is called a Banach *-algebra.
\end{definition} 
\begin{definition}
    A Banach *-algebra $\ca$ is called a $C^*$-algebra if it satisfies the $C^*$-identity: \begin{equation*}
        \|A^*A\|=\|A\|^2
    \end{equation*}
    for all $A\in\ca$.
\end{definition}
A (semi)norm satisfying the $C^*$-identity is called a $C^*$-(semi)norm.
A $C^*$-algebra having an identity element $\mathbf{1}$ is called a \emph{unital} $C^*$-algebra. Even though not every $C^*$-algebra $\ca$ is unital, it is always possible to adjoin an identity $\mathbf{1}$ to obtain a unital $C^*$-algebra $\bar\ca=\mathbb{C}\mathbf{1}+\ca$. The dual of $\ca$, understood as the space of continuous linear functionals $\ca \to \mathbb{C}$, is denoted by $\ca^*$. A natural norm on $\ca^*$ is given by $\|\omega\|=\sup_{\|A\|=1}|\omega(A)|$.
 
Besides the uniform topology, a normed space can be also equipped with other topologies, most importantly for us with the weak-topology. Of central importance here is how convergence is defined with regards to the weak topology:
\begin{definition}
    Let $\ca$ be a normed vector space and let $\{A_\alpha\}_\alpha$ be a net, where every $A_\alpha\in\ca$. We say that that $\{A_\alpha\}_\alpha$ converges weakly to $A\in\a$ if \begin{equation*}
        \lim_\alpha |\om(A_\alpha)-\omega(A)|=0,
    \end{equation*}
    for every $\omega\in\ca^*$.
\end{definition}
Similarly, the dual space $\ca^*$ can be equipped with the weak*-topology, and weak*-convergence is then defined by:
\begin{definition}
Let $\ca$ be a normed vector space. Let $\{\omega_\alpha\}_\alpha$ be a net, where every $\omega_\alpha\in\ca^*$. We say that $\{\omega_\alpha\}_\alpha$ converges in the weak* sense to $\omega\in\ca^*$ if \begin{equation*}
    \lim_\alpha|\omega_\alpha(A)-\omega(A)|=0,
\end{equation*}
for every $A\in\ca$.
\end{definition}
As usual, continuity of a function can be characterized by conservation of limits of nets (nets are required because $\ca^*$ is not typically weak*-metrizable). For normed spaces $\ca$ and $\cb$ and a weakly convergent net  $A_\alpha\rightarrow A$ in $\ca$,
 $f:\ca\rightarrow\cb$, is weakly-continuous exactly when  $f(A_\alpha)\rightarrow f(A)$ in the weak topology on $\cb$. Weak*-continuity for functions $F:\ca^*\rightarrow\cb^*$ can be characterized in a similar way.

An element $A$ of a $C^*$-algebra $\ca$ is positive iff it can be written as $A=B^*B$ for some $B\in\ca$, and a linear functional $\omega\in\ca^*$ is said to be positive if $\omega(A)\geq0$ for all positive elements of $\a$. If $\ca$ is unital and $\omega$ is positive then the norm can be calculated by $\|\omega\|=\omega(\mathbf{1})$. 
\begin{definition}
    A positive linear functional $\omega\in\ca^*$ is called a state if it is normalized i.e.\ $\|\omega\|=1$.
\end{definition}
Positivity of functionals can be used to define a natural ordering on $\ca^*$. We write $\omega_1\geq \omega_2$, i.e.\ $\omega_1$ majorizes $\omega_2$, if $\omega_1-\omega_2$ is positive.
\begin{definition}
    A state $\omega$ is pure if it only majorizes positive linear functionals of the form $\lambda\omega$, with $0\leq\lambda\leq1$. 
\end{definition}
The set of all states (state space) on $\ca$ is denoted as $\cs_\ca$ and the set of all pure states as $\calp_\ca$.
\begin{theorem}
    The state space $\cs_\ca$ is convex. Furthermore, it is weak*-compact iff $\ca$ is unital. In the latter case, the extreme points of $\cs_\ca$ are given by the pure states, and $\cs_\ca$ is the weak*-closure of the convex hull of $\calp_\ca$.
\end{theorem}
For a proof of this theorem, see page 53 below Theorem 2.3.15.\ in~\cite{bratteli1987}.

$C^*$-algebras can be represented as (sub)algebras of the algebra of bounded operators $\cb(\ch)$ on some Hilbert space $\ch$. To make this more precise:
\begin{definition}
    A *-homomorphism between two *-algebras $\ca$ and $\cb$ is a map $\pi:\ca\rightarrow\cb$ such that
    \begin{enumerate}
        \item $\pi(\alpha A+ \beta B)=\alpha\pi(A)+\beta\pi(B)$,
        \item $\pi(AB)=\pi(A)\pi(B)$,
        \item$\pi(A^*)=\pi(A)^*$,
    \end{enumerate}
    for all $A,B\in\ca$ and $\alpha,\beta\in\mathbb{C}$. If a *-homomorphism $\pi$ is bijective, it is called a *-isomorphism.
\end{definition}
\begin{definition}
    A representation of $C$*-algebra A is given by the pair $(\ch,\pi)$, where $\ch$ is a Hilbert space and $\pi:\ca\rightarrow\cb(\ch)$ is a *-homomorphism. If $\pi$ is a *-isomorphism between $\ca$ and the image $\pi[\ca]$, the representation is called faithful.
\end{definition}
As a consequence of the Gelfand–Naimark theorem~\cite{Gelfand1943}, a faithful representation exists for every $C^*$-algebra.

To talk about composite systems, we need the notion of tensor products for $C^*$-algebras. \begin{definition}
Let $\ca$ and $\cb$ be C*-algebras. Their algebraic tensor product is defined by\begin{equation*}
    \a\odot\cb=\left\{\left.\sum_{i=1}^n A_i\otimes B_i\right|A_i\in\ca,B_i\in\cb,n\in\mathbb{N}\right\}.
\end{equation*}
\end{definition}
Although $\ca$ and $\cb$ are $C^*$-algebras, the algebraic tensor product is generally not a $C^*$-algebra. $\ca\odot\cb$ is a *-algebra, but in general it might be possible to define different norms on $\ca\otimes\cb$ satisfying the $C^*$-identity. Furthermore, $\ca\odot\cb$ is not necessarily complete with regard to one of these $C^*$-norms. However, by picking a $C^*$-norm $\|\cdot\|_\alpha$ and completing $\ca\odot\cb$ with respect to it, we obtain the (topological) tensor product $\ca\otimes_\alpha\cb$, which is in fact a $C^*$-algebra itself. Furthermore, a norm $\|\cdot\|_\alpha$ on $\ca\odot\cb$ is called a cross norm if it satisfies $\|A\otimes B\|=\|A\|_\ca\cdot\|B\|_\cb$ for all $A\in\ca$ and $b\in\cb$, where $\|\cdot\|_\ca$ and $\|\cdot\|_\cb$ are the norms on $\ca$ and $\cb$ respectively. It is a well-known fact that every $C^*$-norm on $\ca\odot\cb$ is a cross norm.

In the main text, we will just write $\|\cdot\|$ and $\ca\otimes\cb$ for $\|\cdot\|_\alpha$ and $\ca\otimes_\alpha\cb$, respectively, given that we will fix a particular but arbitrary $C^*$-norm $\|\cdot\|_\alpha$ on $\a\odot\cb$ and consistently work with this choice.

Two well-studied $C^*$-norms are called the minimal and the maximal norm:
\begin{definition}
Let $\ca$ and $\cb$ be C*-algebras and let $(\ch_\ca,\pi_\ca)$ and $(\ch_\cb,\pi_\cb)$ be faithful representations of $\ca$ and $\cb$, respectively. The spatial or minimal norm on $\ca\odot\cb$ is defined as\begin{equation*}
    \left\|\sum_{i=1}^nA_i\otimes B_i\right\|_{\rm{min}}=\left\|\sum_{i=1}^n\pi_{\ca}(A_i)\otimes\pi_{\cb}(B_i)\right\|_{\cb(\ch_\ca\otimes\ch_\cb)}.
\end{equation*}
\end{definition}
It is important to mention that the minimal norm does not depend on the choice of the faithful representations. Note that the minimal norm may be equivalently defined by $\left\|x\right\|_{\rm{min}}=\sup{\left\|\pi_{\ca}\otimes\pi_{\cb}(x)\right\|}$, where the supremum runs now over all representations $\pi_{\ca}$ and $\pi_{\cb}$ of $\ca$ and $\cb$ \cite{takesaki1979}.
\begin{definition}
    The maximal norm on $\ca\odot\cb$ is defined by
    \begin{equation}
    \|C\|_{\mathrm{max}}=\sup\left\{\left.\|C\|_\tau\,|\,\right\|\cdot\|_\tau \mbox{ C*-seminorm on }\ca\odot\cb\right\}.
    \end{equation}
\end{definition}

The $C^*$-algebras obtained by completing $\ca\odot\cb$ with respect to $\|\cdot\|_{\min}$ and $\|\cdot\|_{\max}$ are called the minimal tensor product $\ca\otimes_{\min}\cb$ and the maximal tensor product $\ca\otimes_{\max}\cb$, respectively. 

It turns out that for every $C^*$-norm $\|\cdot\|_\alpha$, it holds that
\begin{equation*}
\|\cdot\|_{\min}\leq\|\cdot\|_\alpha\leq\|\cdot\|_{\max}.
\end{equation*} 
\begin{definition}
    A $C^*$-algebra is called \emph{nuclear} if for every $C^*$-algebra $\cb$, there exists exactly one $C^*$-norm on $\ca\odot\cb$. In other words, $\ca$ is nuclear iff for every $C^*$-algebra $\cb$, the minimal and maximal norm and thus all $C^*$-norms (and all tensor products) coincide.
\end{definition}
Consider two states $\tau$ and $\omega$ on the $C^*$-algebras $\ca$ and $\cb$, respectively. The question arises if  $\tau\otimes\omega$ is a state on $\ca\otimes_\alpha\cb$, and the following proposition answers it in the affirmative.
\begin{proposition}\label{prop:wscs}
    Consider two unital $C^*$-algebras $\ca$ and $\cb$ and states $\tau \in \cs_\ca$ and $\omega \in \cs_\cb$. The linear functional $\tau \otimes \omega$ on
    $\ca \odot \cb$ extends uniquely to a state on $\ca \otimes_{\alpha}\cb$ for any $C^*$-norm $\|\cdot\|_\alpha$.
\end{proposition}
\begin{proof}
    Let $\tau \in \cs_\ca$ and $\omega \in \cs_\cb$ as above. We begin by showing that $|\tau\otimes \omega (u)| \leq \|u\|_{\textrm{min}}$ for any $u \in \ca \odot \cb$. 
    Write $(\mathcal{H}_{\tau},\xi_{\tau},\pi_{\tau})$, $(\mathcal{H}_{\omega},\xi_{\omega},\pi_{\omega})$ for their associated GNS representations, and recall that we may define a representation $\pi_{\tau} \otimes\pi_{\omega}$ of  $\ca\odot\cb$ in $B(\mathcal{H}_{\tau} \otimes \mathcal{H}_{\omega})$ through $\pi_{\tau} \otimes\pi_{\omega} (a \otimes b) = \pi_{\tau}(a) \otimes \pi_{\omega}(b)$, extended by linearity. A short calculation reveals that 
    \begin{equation}
        \tau \otimes \omega (u) = \langle \xi_{\tau} \otimes \xi_{\omega}, \pi_{\tau} \otimes \pi_{\omega} (u) \xi_{\tau} \otimes \xi_{\omega} \rangle ~,
    \end{equation}
    and therefore
    \begin{equation}
        |\tau \otimes \omega (u)| \leq \|\pi_{\tau}\otimes \pi_{\omega}(u)\| \leq \textrm{sup}\|\pi_{\tau}\otimes \pi_{\omega}(u)\| = \|u\|_{\textrm{min}}.
    \end{equation}
    Therefore $\tau \otimes \omega$ is bounded with respect to the minimal norm on $\ca \odot \cb$ and thus extends (uniquely) to $\ca \otimes_{\rm min}\cb$. Normalization is immediate, and positivity follows from writing $\tau \otimes \omega (u) = \| \pi_{\tau}\otimes \pi_{\omega}(x)\xi_{\tau}\otimes \xi_{\omega}\|^2$, where we have written $u=x^*x$. Therefore, $\tau \otimes \omega$ extends to a state. Finally, since for any $C^*$-norm $\|\cdot\|_\alpha$ we have $\|u\|_{\min}\leq \|u\|_{\alpha}$, $\tau \otimes \omega$ extends to a continuous linear functional on any $C^*$-tensor product, and positivity and normalization on $\mathcal{A}\otimes_{\alpha}\mathcal{B}$ follows via continuity from positivity and normalization on the algebraic tensor product $\mathcal{A}\odot\mathcal{B}$.
\end{proof}

We may now address the following which is used in the main text.
\begin{proposition}
    Let $\ca,\cb$ be unital $C^*$-algebras and $\ca \otimes_{\alpha}\cb$ be any $C^*$-tensor product. Then for any fixed $\omega \in \cs_\cb$, the map $\ca^* \to (\ca \otimes_{\alpha}\cb)^*$ defined by $\tau \mapsto \tau \otimes \omega$ is weak*-continuous.
    \label{wstartensorpro}
\end{proposition}
\begin{proof}
We will write $\Phi_{\omega}(\tau):=\tau \otimes \omega$. Let $\tau_{\lambda} \to \tau$ in the weak*-topology. Then we have
$|(\Phi_{\omega}(\tau_{\lambda})-\Phi_{\omega}(\tau))(a \otimes b)| =  |(\tau_{\lambda} - \tau)(a)|\cdot|\omega(b)|$. This vanishes as $\tau_{\lambda} \to \tau$, and by linear extension we have on $\ca \odot \cb$ that $\Phi_{\omega}(\tau_{\lambda} - \tau) (x) \to 0$. Hence we have shown that
$\tau_{\lambda} \otimes \omega (y) \to \tau \otimes \omega (y)$ for all $y$ in the dense subspace $\ca \odot \cb$. Now to show convergence on the completion $\ca \otimes_{\alpha} \cb$, we observe the following. Suppose $x \in \ca \otimes_{\alpha} \cb$ and that $x_{\mu} \to x$ with $x_{\mu} \in \ca \odot \cb$. Now, 
\begin{align*}
    |(\tau_{\lambda} - \tau)\otimes \omega (x)| \leq  |\tau_{\lambda}\otimes \omega(x) - \tau_{\lambda}\otimes \omega(x_{\mu})|
    +  |\tau_{\lambda}\otimes \omega(x_{\mu}) - \tau\otimes \omega (x_{\mu})|
    +  |\tau\otimes \omega(x_{\mu}) - \tau\otimes \omega(x)|.
\end{align*}
The first term on the right hand side vanishes as $x_{\mu} \to x$ by the continuity of each $\tau_{\lambda}\otimes \omega$, the second as $\tau_{\lambda}\otimes \omega \to \tau \otimes \omega$ in the weak*-topology on $(\ca \odot \cb)^*$ by the first part of this proof, and the third as $x_\mu \to x$ by the continuity of $\tau \otimes \omega$.
\end{proof}

In this work, we consider a tripartite scenario $\mathcal{A}\otimes\mathcal{G}\otimes\mathcal{B}$. Some care is required regarding the order of the tensor products. For example, it is well-known that $C^*$-tensor products are not in general associative, i.e.\ there exist $C^*$-algebras $\mathcal{A},\mathcal{B},\mathcal{C}$ and tensor products $\otimes_\alpha,\otimes_\beta$ such that
\[
   (\mathcal{A}\otimes_\alpha\mathcal{B})\otimes_\beta \mathcal{C}\not \simeq \mathcal{A}\otimes_\alpha(\mathcal{B}\otimes_\beta\mathcal{C}),
\]
even though there are canonical isomorphisms if $\alpha$ and $\beta$ are either both min or both max (see e.g.~\cite{BrownOzawa}).
This seems like an obstacle for our construction: going from one time evolution step to the next as in Fig.~\ref{setup}, for example, amounts to a formal ``rebracketing'' of the tripartite tensor product. However, the following lemma is all that is needed to make sense of our results:

\begin{lemma}
\label{LemAssociativitySortOf}
   Let $\ca$, $\cb$, and $\cc$ be unital $C^*$-algebras, and consider $C^*$-tensor products $\ca \otimes _{\alpha}\cb$ and $(\ca \otimes _{\alpha}\cb)\otimes_{\beta} \cc$. Then, there exists uniquely $C^*$-tensor products  $\otimes_\gamma $ and $ \otimes_\delta$ such that $(\ca \otimes _{\alpha}\cb)\otimes_{\beta} \cc \simeq \ca \otimes _{\delta} (\cb \otimes_{\gamma}\cc)$ under the canonical rebracketing isomorphism $a\otimes (b\otimes c)\mapsto (a\otimes b)\otimes c$.
\end{lemma}
\begin{proof}
Call  $( \ca \otimes _{\alpha}\cb)\otimes_{\beta} \cc:= \cd$ with norm $\|\cdot \|_\cd$. The 
 map $j:\mathcal{B}\odot\mathcal{C}\to\mathcal{D}$, $b \otimes c \mapsto (1 \otimes b) \otimes c$ (extended linearly) is an injective $*$-homomorphism. On $\mathcal{B}\odot\mathcal{C}$, we define the norm $\|x\|_\gamma:= \|j(x) \|_{\cd}$, which makes $j$ isometric.
 That the $C^*$-identity holds follows from direct computation: $\|x^* x\|_{\gamma}=\|j(x^* x)\|_{\mathcal{D}}=\|j(x)^* j(x)\|_{\mathcal{D}}=\|j(x)\|^2_\mathcal{D}=\|x\|_\gamma^2$, where in the penultimate equality we have used the $C^*$-identity in $\cd$.
 This gives us a $C^*$-tensor product $\mathcal{B}\otimes_\gamma \mathcal{C}$. Next let $\mathcal{T}:=\mathcal{A}\odot(\mathcal{B}\odot\mathcal{C})$, and consider the *-homomorphism $r:\mathcal{T}\to\mathcal{D}$, defined by linear extension of $a\otimes (b\otimes c)\mapsto (a\otimes b)\otimes c$. On $\mathcal{T}$, define the norm $\|y\|_\delta:=\|r(y)\|_\mathcal{D}$, and let $\mathcal{E}$ be the completion of $\mathcal{T}$ under this norm. Now, $\mathbf{1}\odot(\mathcal{B}\odot\mathcal{C})$ is a linear subspace of $\mathcal{T}$, and if $x\in\mathcal{B}\odot\mathcal{C}$ then
 \[
    \|\mathbf{1}\otimes x\|_\delta=\|r(\mathbf{1}\otimes x)\|_\mathcal{D}=\|j(x)\|_\mathcal{D}=\|x\|_\gamma.
 \]
 Hence the closure of this subspace is isometrically isomorphic to $\mathcal{B}\otimes_\gamma\mathcal{C}$. In this sense, $\mathcal{A}\odot(\mathcal{B}\otimes_\gamma\mathcal{C})$ is a linear subspace of $\mathcal{E}$, inheriting the norm $\|\cdot\|_\delta$ from $\mathcal{E}$. On $\mathcal{T}$, the norm $\|\cdot\|_\delta$ has the $C^*$-property due to the fact that $r$ is a *-homomorphism; by continuity, this is also true on $\mathcal{A}\odot(\mathcal{B}\otimes_\gamma\mathcal{C})$, and hence $\mathcal{E}\simeq \mathcal{A}\otimes_\delta(\mathcal{B}\otimes_\gamma\mathcal{C})$. The map $r$ can be uniquely isometrically extended to all of $\mathcal{E}$; by continuity, the resulting map will still be a *-homomorphism, establishing the isomorphism between $\mathcal{E}$ and $\mathcal{D}$.
%
%
Uniqueness follows from the fact that 
the algebraic tensor product $\mathcal{A}\odot(\mathcal{B}\odot\mathcal{C})$ is dense in every triple $C^*$ product $\mathcal{A}\otimes_\delta(\mathcal{B}\otimes_\gamma\mathcal{C})$, and demanding that the rebracketing map extends to an isometric isomorphism fixes the norm on this dense subspace uniquely to be the corresponding norm from $(\mathcal{A}\otimes_\alpha\mathcal{B})\otimes_\beta\mathcal{C}$.
\end{proof}

In the main text, we assume that we always deal with local interactions that can be extended to interactions on the total system; that is, we assume that we only deal with completely positive maps $T_{\ca\cg}:\ca\otimes\cg\rightarrow\ca\otimes\cg$ and $\phi_\cb:\cb\to\cb$ such that $T_{\ca\cg}\otimes\phi_B:\ca\otimes\cg\otimes \cb \to \ca\otimes\cg\otimes\cb$ exists and is completely positive. We justify this assumption by the requirement that we only need to consider processes that are physically implementable. One could ask if this assumption is actually needed or whether the combination of two completely positive maps can always be extended to a completely positive map on any tensor product of their underlying $C^*$-algebras. We will see in the following example that the assumption is actually needed.
\begin{example}[Nonexistence of completely positive extension]
\label{ExNonexistence}
Let $\cd$ be a $C^*$-algebra such that $\cd\otimes_{\min} \cd \neq \cd\otimes_{\max}\cd$ (for example, $\cd=\mathcal{B}(\mathcal{H})$ for $\mathcal{H}=\ell^2$~\cite{JungePis}). We consider a sequence $\{u_n\}_{n\in\mathbb{N}}\subset\cd\odot\cd$ s.t. $\|u_n\|_{\min}\to 0$ but $\|u_n\|_{\max}\not\to 0$. Such a sequence exists since $\|\cdot\|_{\min}\leq\|\cdot\|_{\max}$ and we assumed $\cd\otimes_{\min}\cd\neq \cd\otimes_{\max}\cd$, which implies that $\|\cdot\|_{\min}$ and $\|\cdot\|_{\max}$ are inequivalent norms. Furthermore, let $\ca=\cd\oplus\cd$ and $\pi_1:\ca\rightarrow D$ be the projection to the first entry, i.e.\ $\pi_1(a_1,a_2)=a_1$. Let us consider the norm $\|\cdot\|_\alpha$ on $\ca\odot\cd$~\cite{Ozawa2016} defined by
\begin{equation*}
    \|z\|_\alpha=\max\left\{\|z\|_{\min},\|(\pi_1\otimes\mathrm{id})(z)\|_{\max}\right\}.
\end{equation*}
This is in fact a $C^*$-norm, which follows from the following. Since $\pi_1 \otimes \mathrm{id}: \ca \odot \cd \to \cd \odot \cd$ is a $*$-homomorphism, we have $\|z^*z\|_\alpha=\max\left\{\|z\|^2_{\min},\|(\pi_1\otimes\mathrm{id})(z)\|^2_{\max}\right\} = \max\left\{\|z\|_{\min},\|(\pi_1\otimes\mathrm{id})(z)\|_{\max}\right\}^2 = \|z\|_{\alpha}^2$, so that the $C^*$-identity holds for this norm. Now, we consider the map $\Psi:\ca\rightarrow \a$:
\begin{equation}
   \Psi(a_1,a_2)=(a_2,0), 
\end{equation}
which is a *-homomorphism and thus completely positive~\cite{Paulsen2003}. We will see that $\Psi\otimes\mathrm{id}:\ca\odot\cd\rightarrow\ca\odot\cd$ cannot be extended to a completely positive map on $\ca\otimes_\alpha\cd$. Let $\iota_2: \cd\rightarrow\ca$ be the embedding $\iota_2(a)=(0,a)$ and let $x_n:=(\iota_2\otimes\mathrm{id})(u_n)\in\mathcal{A}\odot\mathcal{D}$. We calculate
\begin{equation*}
 (\pi_1\otimes\mathrm{id})(x_n)=((\pi_1\circ \iota_2)\otimes \mathrm{id})(u_n)=0
\end{equation*}
and hence
\begin{equation}
    \|x_n\|_{\alpha}=\max\{\|x_n\|_{\min},0\}=\|x_n\|_{\min}=\|u_n\|_{\min}. \label{normxneqnormun}
\end{equation}
To see that the last equality is true, let $\pi_{\ca}$ be a faithful representation of $\mathcal{A}$. Since $\iota_2$ is an injective *-homomorphism, $\pi_\cd=\pi_{\ca}\circ\iota_2$ is a faithful representation of $\cd$.  Now, let $x=\sum_{k=1}^n (0,a_k)\otimes b_k\in \ca\odot\cd$, $u=\sum_{k=1}^n a_k\otimes b_k\in \cd\odot\cd$ and $\tilde{\pi}$ be a faithful representation of $\cd$, we find $\|x\|_{\min}=\|\sum_{k=1}^n\pi_{\ca}((0,a_k))\otimes\tilde{\pi}(b_k)\|=\|\sum_{k=1}^n\pi_{\ca}(\iota_2(a_k))\otimes\tilde{\pi}(b_k)\|=\|\sum_{k=1}^n\pi_{\cd}(a_k)\otimes\tilde{\pi}(b_k)\|=\|u\|_{\min}$.
From (\ref{normxneqnormun}), we conclude $\|x_n\|_{\alpha}\to 0$. The map $\Psi\otimes\mathrm{id}:\ca\odot\cd\rightarrow\ca\odot\cd$ is linear and under the assumption that $\Psi\otimes \mathrm{id}:\ca\otimes_\alpha\cd\rightarrow \ca\otimes_\alpha\cd$ exists and is completely positive and thus continuous, from $\|x_n\|_\alpha\to 0$ it follows $\|\Psi\otimes\mathrm{id}(x_n)\|_{\alpha}\to 0$. However, we will show that this leads to a contradiction. To this end, \begin{equation*}
    \Psi\circ\iota_2(a)=\Psi(0,a)=(a,0)=\iota_1(a),
\end{equation*}
where we have defined the embedding $\iota_1$ in a similar way as $\iota_2$. Furthermore,
\begin{equation*}
    \Psi\otimes\mathrm{id}(x_n)=\Psi\otimes\mathrm{id}(\iota_2\otimes\mathrm{id})(u_n)=\iota_1\otimes\mathrm{id}(u_n)
\end{equation*}
and
\begin{equation*}
    (\pi_1\otimes\mathrm{id})((\Psi\otimes\mathrm{id})(x_n))=((\pi_1\circ\iota_1)\otimes\mathrm{id})(u_n)=(\mathrm{id}\otimes\mathrm{id})(u_n)=u_n.
\end{equation*}
Thus\begin{equation*}
    \|\Psi\otimes\mathrm{id}(x_n)\|_{\alpha}\geq \| (\pi_1\otimes\mathrm{id})((\Psi\otimes\mathrm{id})(x_n))\|_{\max}=\|u_n\|_{\max},
\end{equation*}
but therefore $\|\Psi\otimes\mathrm{id}(x_n)\|_\alpha\not\to 0$, which is a contradiction.
\end{example}

\section{On LOCC versus classical mediators in the formulation of GME}

The GME proposals in the literature model the classical gravitational field in two distinct manners. In the approach of~\cite{bose2017} the field is modeled as a classical communication channel between the two systems. In the approach of~\cite{marletto2017}, and later~\cite{galley2022} the gravitational field is modeled as a classical system which acts as a mediator, as in the present work. We call this approach the mediator approach, in contrast to the LOCC (local operations and classical communication)~\cite{Horedecki2009,chitambar2014} approach.

In~\cite{marletto2017,marletto2017a}  the classical field is a single bit, whereas in~\cite{galley2022}, the classical field is a discrete classical system (hence with a finite-dimensional space of probability distributions) of arbitrary size. Thus, the present work significantly generalizes the classical mediators used in~\cite{marletto2017,marletto2017a,galley2022} from bits or other discrete classical systems to commutative unital $C^*$-algebras. These algebras can also describe discrete classical systems, but also much more general classical systems such as Hamiltonian mechanics on finite-dimensional phase spaces~\cite{Duvenhage}. They are also expected to describe classical systems on infinite-dimensional phase spaces such as classical field theories, even though this is to the best of our knowledge not spelled out explicitly in the literature.

Furthermore, in contrast to the LOCC approach, our formulation also allows the local quantum systems $A$ and $B$ to be infinite-dimensional, or, in fact, to be described by any unital $C^*$-algebra. Hence, while the LOCC approach treats $A$ and $B$ as quantum-information-theoretic qudits, our results apply to the cases where, for example, $A$ and $B$ are collections of harmonic oscillators or (subsystems of) quantum fields.

Unlike in the mediator approach,  the classical system in the LOCC approach used in~\cite{bose2017} is not modeled dynamically  and hence is not explicitly characterized.  It enters via a communication channel which can be used to transmit information from Alice to Bob, and reciprocally, over a number of rounds. The `interaction' of the classical system with the local quantum systems is modeled via the conditioning of the choice of instrument of one party based on the outcomes of the previous instruments.

The only property of the classical system in the LOCC approach is its dimension (which determines how much classical information can be communicated). In order to properly contrast our contribution based on the mediator approach to the existing literature based on the LOCC approach, we show how to translate the LOCC protocol into a classical mediator protocol, and compare the classical systems used in both approaches. This allows us to compare the dimension of the classical mediator used in the LOCC and mediator approaches.

\subsection{LOCC protocols}

In the following, we outline the definition of two-party LOCC operations. For more details, see e.g.~\cite{chitambar2014}. We note that while there exist generalizations of the definition of LOCC to the case of von Neumann algebras~\cite{VL3} and of $C^*$-algebras~\cite{Verch2005}, we restrict our attention to the definition commonly used in quantum information theory and in the GME proposal of~\cite{bose2017}.

Given two parties $\A$ and $\B$ with associated Hilbert space decomposition $\ch_\A \otimes \ch_\B$, where $\ch_\A$ and $\ch_\B$ are finite-dimensional, a local instrument for $\A$ is of the form $\{M_i^\A \otimes \I_\B\}_{i \in \ci}$, and similarly, a local instrument for Bob is of the form $\{\I_\A \otimes N_i^\B\}_{i \in \ci}$, where $M_i^\A$ and  $N_i^\B$ are sub-normalised completely positive maps and $\I_\A$/$\I_\B$ is the identity channel. In the following, we leave implicit the identity operators on tensor factors, writing for instance $\{M_i^\A\}_{i \in \ci}$ for $\{M_i^\A \otimes \I_\B\}_{i \in \ci}$. The outcome set $\ci$ of the instruments is a countable set, which may be finite or infinite. Without loss of generality it is assumed to be fixed for a given LOCC protocol.  

In the first round of an LOCC protocol, Alice implements an instrument $\{M_i^\A\}_{i \in \ci}$ and communicates the outcome $i$ to Bob, who  implements an instrument $\{N_{j|i}^\B\}_{j \in \cj}$ depending on the outcome $i$. The action of these two instruments for outcome $(i,j)$ on a state $\rho \in \cb(\ch_\A \otimes \ch_\B)$ is:
\begin{align}
    \rho \mapsto (\I_\A \otimes N_{j|i}^\B) [(M_i^\A \otimes \I_\B)\rho)] 
\end{align}
If one coarse grains over Bob's possible outcomes $j$, one obtains the following:
\begin{align}
     \rho \mapsto (\I_\A \otimes T_i^\B) [(M_i^\A \otimes \I_\B)\rho)]
\end{align}
where $T_i^\B = \sum_j N_{j|i}^\B$ is now a channel. The instrument $J_1 = \{(T_i^\B \otimes  M_i^\A)\}_{i\in \ci}$ is a \emph{one way local LOCC} instrument from Alice to Bob, with corresponding channel given by the coarse graining $T_1 = \sum_{i} (M_i^\A \otimes T_i^\B)$.

In the next round, Bob communicates the outcome $j$ of his instrument $\{N_{j|i}^\B\}_{j \in \cj}$ to Alice, who implements her own instrument $\{O_{k|i,j}^\A\}_{k \in \ck}$ conditional on the outcomes $i$ and $j$, giving the following action for outcome $(i,j,k)$:
\begin{align}
    \rho \mapsto  (O_{k|i,j}^\A \otimes \I_\B) (\I_\A \otimes N_{j|i}^\B) [(M_i^\A \otimes \I_\B)\rho)].
\end{align}
If one coarse grains over Alice's outcomes $k$, one obtains for outcome $(i,j)$:
\begin{align}
    \rho \mapsto \sum_k (O_{k|i,j}^\A \otimes \I_\B) (\I_\A \otimes N_{j|i}^\B) [(M_i^\A \otimes \I_\B)\rho)] = (T^\A_{ij}M_i^\A \otimes N_{j|i}^\B)(\rho) .
\end{align}
Coarse graining over the outcomes $i$ gives  $ \{\sum_i T^\A_{ij}M_i^\A \otimes N_{j|i}^\B \}_{j}$ which is a \emph{two-way LOCC instrument}, with corresponding channel given by:
\begin{align}
     \rho\mapsto\sum_{i,j,k} (O_{k|i,j}^\A \otimes \I_\B) (\I\A \otimes N_{j|i}^\B) [(M_i^\A \otimes \I_\B)\rho)] =  \sum_{i,j} (T_{i,j}^\A \otimes \I_\B) (\I\A \otimes N_{j|i}^\B) [(M_i^A \otimes \I_\B)\rho)].
\end{align}
Observe that for every $i$, the instrument $\{T_{i,j}^\A \circ N_{j|i}^\B\}_{j \in \cj}$ is a one-way local LOCC instrument from Bob to Alice. The collection of CP maps $\{(T_{i,j}^\A \circ N_{j|i}^\B \circ M_i^\A)\}_{i \in \ci, j \in \cj}$ can be coarse grained to define an instrument $\{\sum_i T_{i,j}^\A \circ N_{j|i}^\B \circ M_i^\A \}_{j \in \cj}$ which is \emph{LOCC linked} to  $\{(T_i^\B \circ M_i^\A)\}_{i \in \ci}$.

The above consists of two rounds of an LOCC protocol, and generalizes to $n$-round protocols as well as infinite-round protocols by means of a limit procedure described below.

For details of the definition, we refer the reader to~\cite{chitambar2014}. Here, we are interested in the question of state transitions by LOCC channels; hence, it is the channels (and not the instruments) that are our subject of main interest, and what are to be simulated by classical mediators, as discussed in the next subsection.

Observe that while Alice's choice of instrument $\{O_{k|i,j}\}_k$ depends on both $i$ and $j$, Bob needs only transmit the outcome $j$, since it is assumed that Alice can locally store the outcome $i$ from her initial instrument. In the LOCC paradigm agents have free access to local classical systems. This is in contrast with the classical mediator paradigm where all the classical systems have to be included in the mediator.

\subsection{LOCC protocols using a classical mediator (finite number of rounds)}

In the following we establish that finite round LOCC protocols can be implemented by two parties sharing a classical mediator, similarly as used in~\cite{jensen2020} for finite mediators. 
In the following, we denote the classical system $\ell^\infty$ of~\Cref{ex:inf_classical_system} by $\Delta_\infty$ and the finite-dimensional classical system over $n$ points by $\Delta_n$.

\begin{lemma}
    Consider an LOCC protocol between two parties $A$ and $B$, with a classical communication channel of (finite or countably-infinite) dimension $n$ that is used $m$ times. Then the resulting LOCC quantum channel can be implemented in the setup of Figure~\ref{setup} (which shows two steps) with a classical mediator $\Delta_{mn}$ in $m$ steps.
\end{lemma}

\begin{proof}

We first consider the case where the classical communication channel is used to transmit a classical system of finite dimension $n$.

The local quantum systems are $\ch_\A \otimes \ch_\B$ together with a classical system $\C = \C_{\A_1} \C_{\B_2} ... \C_{\A_m}$ for $m$ odd and $\C = \C_{\A_1} \C_{\B_2} ... \C_{\B_m}$ for $m$ even. Each system $\C_{\A_i}$ and $\C_{\B_j}$ is of dimension $n$, hence $\C$ is of dimension $mn$.

Consider the first use of the communication channel from Alice to Bob:
    \begin{align}
     \rho \mapsto \sum_i (\I_\A \otimes N_{j|i}^\B) [(M_i^\A \otimes \I_\B)\rho)]
\end{align}
In the classical mediator scenario, Alice performs the same instrument $\{M_i^\A\}_{i \in \ci}$, and conditional on the outcome prepares the mediator in the state $\ketbra{i}{i}_{\C_{\A_1}}$:
\begin{align}
    \rho \mapsto \rho'= \sum_i (M_i^\A \otimes \I_\B)\rho \otimes \ketbra{i}{i}_{\C_{\A_1}}
\end{align}
This is one application of a classical-quantum channel on the systems $\A\C_{\A_1}$ and hence one  step in the classical mediator setup.

Bob then measures the classical mediator in the canonical basis and implements his instrument $\{N_{j|i}^\B\}_{j \in \cj}$ conditional on the outcome of the measurement. Similarly conditioned on the outcome of his instrument, he prepares the state $\ketbra{j}{j}_{C_{B_2}}$:
\begin{align}
    \rho' \mapsto \rho' \! ' &= \sum_{j,k} \ketbra{j}{j}_{\C_{\B_2}} \otimes (N_{j|k}^\B \otimes E_{\ketbra{k}{k}_{\C_{A_1}}})(\rho') \\
    &=  \sum_{i,j} (\I_\A \otimes N_{j|i}^\B) [(M_i^\A \otimes \I_\B)\rho)] \otimes \ketbra{j}{j}_{\C_{\B_2}} \otimes \ketbra{i}{i}_{\C_{\A_1}}  .
\end{align}
where $E_{\ketbra{k}{k}_{\C_{\A_1}}}(\cdot) = \ketbra{k}{k}_{\C_{\A_1}} (\cdot) \ketbra{k}{k}_{\C_{\A_1}}$.

This consists of an application of a classical quantum channel on the systems $\B\C_{\A_1} \C_{\B_2}$ in the mediator setup. By tracing out the system $\C_{\A_1} \C_{\B_1}$ one recovers the same channel that is implemented by the one-way LOCC protocol. We note that if there are no further rounds Bob does not need to encode his outcome $j$ in the system $\C_{\B_1}$, hence the one-way LOCC channel can be implemented in a single step with the classical system $\C_{\A_1}$ in the mediator setup.

It is now clear how this generalizes to more rounds. For the subsequent use of the communication channel, Alice measures the classical system $\C_{\A_1} \C_{\B_2}$ and implements the instrument $O_{k|i,j}$ depending on the outcome $i,j$. She then encodes the outcome $k$ in a further classical system $\C_{\A_3}$ which she would (in the LOCC case) send to Bob. Hence, each use of the communication channel requires an additional classical system $\Delta_n$ to be used by each agent.

Now consider the case of a countable-infinite classical message (an integer) to be sent during each use of the communication channel in the LOCC protocol. Let $S$ be the set of all finite sequences of integers, then $S$ is countable. Hence there is a bijective map $f:S\to\mathbb{N}$. After $k$ steps of the LOCC protocol, the sequence of previous measurement outcomes $(n_1,\ldots,n_k)\in S$ will be stored in the state $|f(n_1,\ldots,n_k)\rangle\langle f(n_1,\ldots,n_k)|$ of the classical mediator. The next party measures the mediator and, using $f^{-1}$, determines $(n_1,\ldots,n_k)$, performs its operation $O_{n_{k+1}|n_1,\ldots,n_l}$ and reprepares the classical system in state $|f(n_1,\ldots,n_k,n_{k+1})\rangle\langle f(n_1,\ldots,n_k,n_{k+1})|$. This implements the corresponding LOCC channel with a classical mediating system $\Delta_\infty=\ell^\infty$.
\end{proof}

\subsection{Infinite-round LOCC instruments and channels}

It is well-known that there are quantum instruments that are not inplementable exactly by any finite-round LOCC protocol, but that can nonetheless be approximated to arbitrary accuracy by such LOCC protocols. There are two versions of this, called ${\rm LOCC}$ and $\overline{{\rm LOCC}_{\mathbb{N}}}$ in~\cite{chitambar2014}: \textit{``[...] for any instrument in ${\rm LOCC}$, its approximation in finite rounds can be made tighter by just continuing for more rounds within a \emph{fixed} LOCC protocol; whereas for instruments in $\overline{{\rm LOCC}_{\mathbb{N}}}\setminus {\rm LOCC}$, different protocols will be needed for different degrees of approximation.''}

However, all known examples of such instruments are proper instruments in the sense that they have a non-trivial classical outcome set. In the special case of channels (which is our case of interest), it seems currently unknown whether a finite number of LOCC rounds is always sufficient for every channel in $\overline{{\rm LOCC}_{\mathbb{N}}}$ to be implemented perfectly~\cite{Cohen2024}. A positive answer to this would mean that the set of channels in ${\rm LOCC}_{\mathbb{N}}$, the finite-round LOCC channels, is topologically closed.

If it turns out that this set is not closed, the channels in $\overline{{\rm LOCC}_{\mathbb{N}}}\setminus {\rm LOCC}_{\mathbb{N}}$ may or may not be implementable exactly with a classical mediator described by some commutative $C^*$-algebra; but they can certainly be implemented \textit{to arbitrary accuracy} with the classical mediator $\Delta_\infty=\ell^\infty$. In this sense, the classical mediator approach contains the LOCC approach as a special case, with the limiting channels of the former containing the limiting channels of the latter. Hence, the state transitions implementable in the limit of infinitely many steps by the LOCC approach can all be implemented in the limit of infinitely many steps by the classical mediator approach. Moreover, in the main text, we have shown that even the infinite-step limit of the classical mediator setup cannot entangle $A$ and $B$.

\end{document}